\documentclass[runningheads,envcountsame,a4paper]{llncs}

\usepackage{microtype}

\usepackage{amsmath}
\usepackage{amssymb}
\usepackage{mathtools}
\usepackage{verbatim} 
\usepackage{stmaryrd}
\usepackage{url}
\allowdisplaybreaks
\usepackage{basic}
\usepackage{ambients}
\usepackage{paralist}
\usepackage{graphicx}
\usepackage{proof}
\usepackage{hyperref}

\usepackage{enumitem}

\usepackage{epsf}
\usepackage{graphics} % for pdf, bitmapped graphics files
\usepackage{epsfig} % for postscript graphics files

\newcommand{\env}{\mathit{Env}}

\newif\ifdraft\drafttrue
\newcommand{\ntrans}[1]{\mathrel{{\trans{#1}}\makebox[0em][r]{$\not$\hspace{2ex}}}{\!}}
\newcommand{\ttrans}[1]{\stackrel{\, {#1} \,}{\Longrightarrow}}
\newcommand{\ttranst}[1]{\Longrightarrow\trans{#1}\Longrightarrow}

\newcommand{\rel}{\mathcal R}

\newcommand{\on}{\mathsf{on}}
\newcommand{\off}{\mathsf{off}}

\newcommand{\confCPS}[2]{#1 \, {\Join} \, #2}

\newcommand{\rsens}[2]{\mathsf{read}\, #2(#1)}

\newcommand{\wact}[2]{\mathsf{write}\, #2 \langle #1 \rangle}

\newcommand{\statefun}{\xi_{\mathrm{x}}}
\newcommand{\actuatorfun}{\xi_{\mathrm{u}}}
\newcommand{\uncertaintyfun}{\xi_{\mathrm{w}}}
\newcommand{\evolmap}{\mathit{evol}}
\newcommand{\errorfun}{\xi_{\mathrm{e}}}
\newcommand{\measmap}{\mathit{meas}}
\newcommand{\invariantfun}{\mathit{inv}}
\renewcommand{\operatorname}[1]{\mathit{#1}}

\newcommand{\CPS}{CPS}
\newcommand\restrict[1]{\raise-.5ex\hbox{\ensuremath|}_{#1}}
\usepackage{marvosym} 

\begin{document}

\title{A Calculus of Cyber-Physical Systems\thanks{An extended abstract appeared in the Proc.\ of LATA 2017, volume 10168 of Lecture Notes in Computer Science, pp.\ 115-127, Springer, 2017.}}
\titlerunning{A Calculus of Cyber-Physical Systems}

\author{Ruggero Lanotte\inst{1} \and Massimo Merro\inst{2}}
\authorrunning{R.\ Lanotte and M.\ Merro}

\institute{Dipartimento di Scienza e Alta Tecnologia, Universit\`a dell'Insubria, Como, Italy \and Dipartimento di Informatica, Universit\`a degli Studi di Verona, Italy}

\maketitle

\begin{abstract}
We propose a hybrid process calculus for modelling  and reasoning on
\emph{cyber-physical systems\/} (\CPS{s}). 
The dynamics of the calculus is expressed in terms of 
a \emph{labelled transition system\/} in the SOS style of Plotkin. 
This is used to define a 
\emph{bisimulation-based\/} behavioural semantics
which support compositional reasonings. Finally, we prove run-time properties and system equalities for a non-trivial case study. 
\keywords{Process calculus, cyber-physical system, semantics.}
\end{abstract}

%%%%%%%%%%%%%%%%%%%%%%%%%%%%%%%%
%%%%%%                                                                      %%%%%%%
%%%%%%                 I N T R O D U C T I O N                %%%%%%%
%%%%%%                                                                      %%%%%%%
%%%%%%%%%%%%%%%%%%%%%%%%%%%%%%%%

\section{Introduction}

%%\enlargethispage{.2\baselineskip}
\emph{Cyber-Physical Systems} (\CPS{s}) are integrations of 
networking and distributed computing systems with physical processes,
% that
%%monitor and control entities in the physical environment,
 where feedback loops allow physical processes to affect computations and vice versa. For
example, in real-time control systems, a hierarchy of \emph{sensors},
\emph{actuators} and \emph{control processing components} are connected to
control stations. Different kinds of \CPS{s} include
supervisory control and data acquisition (SCADA), programmable logic
controllers (PLC) and distributed control systems.

\begin{figure}[t]
\centering
\includegraphics[width=6.5cm,keepaspectratio=true,angle=0]{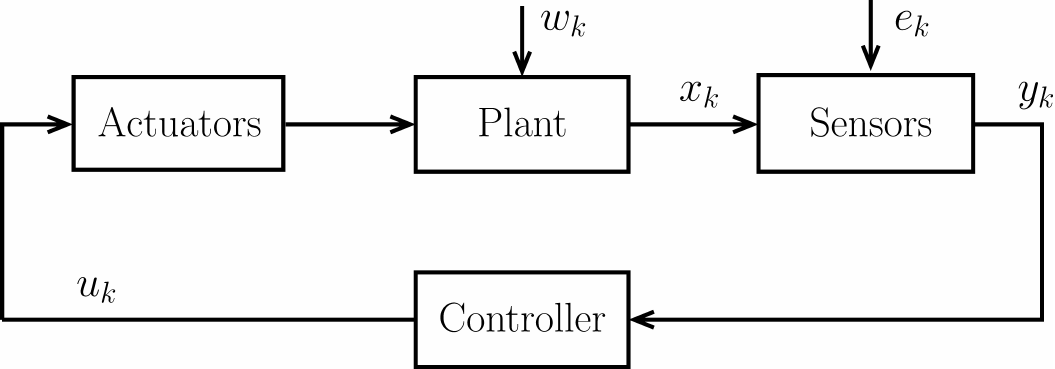}
\caption{Structure of a \CPS}
\label{fig:cps-model}
\end{figure}

The \emph{physical plant} of a \CPS{} is typically 
represented by means of a \emph{discrete-time state-space
model\/}\footnote{See~\cite{survey-CPS-security-2016} for a tassonomy of time-scale models used to represent \CPS{s}.} consisting of two 
equations of the form
\begin{center}
\begin{math}
\begin{array}{rcl}
x_{k+1} & = & Ax_{k} + Bu_{k} + w_{k}\\[2pt]
y_k & = & Cx_{k} + e_k
\end{array}
\end{math}
\end{center}
where
$x_k \in \mathbb{R}^n$ is the current \emph{(physical) state}, $u_k \in
\mathbb{R}^m$ is the \emph{input} (i.e., the control actions implemented
through actuators) and $y_k \in \mathbb{R}^p$ is the \emph{output} (i.e.,
the measurements from the sensors).
The \emph{uncertainty} $w_k \in  \mathbb{R}^n$ and the \emph{measurement error} $e_k \in  \mathbb{R}^p$ represent perturbation and sensor noise, 
respectively, and $A$, $B$, and $C$ are matrices modelling the dynamics of the physical system.  
%%This type of state space model can be used to describe the dynamic behavior of a physical process, also known as \emph{physical plant}, where 
The \emph{next state} $x_{k+1}$ depends on the current state $x_k$ and 
the corresponding control actions $u_k$, at the sampling instant $k
\in \mathbb{N}$. Note that, the state $x_k$ cannot be directly observed: only its
measurements $y_k$ can be observed.

The physical plant is supported by a communication network through which
the sensor measurements and actuator data are exchanged with the
\emph{controller(s)\/}, i.e., the \emph{cyber} component, also called
logics,  of a CPS (see \autoref{fig:cps-model}).

%%Applications of CPS arguably have the potential to dwarf the 20-th century IT revolution~\cite{CPS-foundations}. 
The range of \CPS{s} applications is rapidly increasing and 
already covers several domains~\cite{CPS-applications}: 
%%high confidence medical devices and systems, %%assisted living,
%% traffic control and safety, 
advanced automotive systems, %%process control, 
energy conservation, environmental monitoring, avionics, %% instrumentation, 
critical infrastructure control (electric power, water resources, and communications systems for example), 
%%distributed robotics (telepresence, telemedicine), %%defense,
%% manufacturing, %%smart structures, 
etc. 

However, there is still a lack of research on the modelling and validation of \CPS{s} through formal methodologies that might allow to model the interactions among the system components, and to verify the correctness of a \CPS{}, as a whole, before its practical implementation. A straightforward utilisation of these techniques is for \emph{model-checking}, i.e.\ to statically assess whether the current system deployment can guarantee the expected behaviour. However, they can also be an important aid for system planning, for instance to decide whether 
 different deployments for a given application are  behavioural equivalent.

In this paper, we propose a contribution in the area of \emph{formal methods} for \CPS{s}, by defining  a \emph{hybrid process calculus}, called \cname{}, with a clearly-defined \emph{behavioural semantics\/} for specifying and reasoning on \CPS{s}.  
%%Process calculi have been successfully used to model distributed and mobile systems (see, e.g., the $\pi$-calculus~\cite{Mil91}). 
%%\footnote{Process calculi have been successfully used to model and analyse distributed and mobile systems as well as cryptographic protocols
%%(see, e.g., the \emph{$\pi$-calculus}~\cite{Mil91},
%%\emph{Ambients}~\cite{Ambients} and the \emph{Spi Calculus}~\cite{spi}).}
%%%However, to better describe systems based on a particular paradigm, dedicated calculi are needed. \cname\  targets systems based on the CPS
%%%paradigm. 
In \cname{}, systems are represented as terms of the form 
%\begin{math}
$\confCPS E  P $, 
%\end{math}
where $E$ denotes the \emph{physical plant} (also called environment) of the
system, containing information on state variables, actuators, sensors,
evolution law, etc., while $P$ represents the \emph{cyber
component of the system}, i.e., the \emph{controller} that governs sensor reading and
actuator writing, as well as channel-based communication with other cyber
components. Thus, channels are used for logical interactions between cyber
components, whereas sensors and actuators make possible the interaction
between cyber and physical components. 
Despite this conceptual similarity, 
messages transmitted via channels are ``consumed'' upon reception,  
whereas  actuators' states (think of a valve)
remains unchanged until its controller modifies it. 
%%Similarly, a 
%%sensor value changes only when the corresponding state variable
%%changes, depending on the evolution law of the physical environment $E$.

\enlargethispage{\baselineskip}
%%
%%The cyber components  of our \CPS{s} are formally defined via a process calculus that builds on CCS with \emph{discrete time evolution}~\cite{HR95} enriched with specific constructs to model the interaction with sensors and actuators.

%\enlargethispage{\baselineskip}
 \cname{}   is equipped with a \emph{labelled transition semantics}
(LTS) in the SOS style of Plotkin~\cite{Plo04}. 
 We prove that our labelled transition semantics satisfies 
some standard time properties such as: \emph{time determinism\/}, 
\emph{patience}, \emph{maximal progress}, and \emph{well-timedness\/}.
Based on our LTS, we define a natural notion of \emph{weak bisimilarity}. 
As a main result, we prove that our bisimilarity is a congruence 
and it is hence suitable for \emph{compositional reasoning\/}.
We are not aware of similar results in the context of  \CPS{s}. 
% \marginpar{MM: Non so se sia il caso di parlare di deadlock.
% Perche' la trace in genere non osserva deadlock}
%bisimulation-based system equalities
Finally, we provide a non-trivial \emph{case study}, taken from an
engineering application, and use it to illustrate our definitions and 
our semantic theory for \CPS{s}. 
 Here, we wish to
%%We 
remark that while we have kept the example simple, it is actually far from trivial and designed to show that various \CPS{s} can be modelled in this style.
\paragraph*{Outline}
In \autoref{sec:calculus}, we give syntax and 
 operational semantics of
 \cname{}. In \autoref{sec:bisimulation} we provide a bisimulation-based 
behavioural semantics for \cname{} and prove its compositionality. 
In \autoref{sec:case-study} we model in \cname{} our case study, and prove for it 
 run-time properties as well as system equalities. In \autoref{sec:conclusions}, we discuss related and future work.

%%%%%%%%%%%%%%%%%%%%%%%%%%%%%%%%
%%%%%%                                                                      %%%%%%%
%%%%%%                        A L G E B R A                         %%%%%%%
%%%%%%                                                                      %%%%%%%
%%%%%%%%%%%%%%%%%%%%%%%%%%%%%%%%

\section{The Calculus}
\label{sec:calculus}

In this section, we introduce our \emph{Calculus of Cyber-Physical
Systems} \cname{}. 
Let us start with some preliminary notations. 
%% in order to define our 
%%calculus \cname{}.
We use   $x, x_k \in \cal X$ for \emph{state variables\/};   
%%(associated to physical states of systems), 
 $c,d \in \cal C$ for \emph{communication channels\/}, 
 $a, a_k \in \cal A$ for \emph{actuator devices\/}, 
 $s,s_k \in \cal S$ for \emph{sensors devices\/}.
%%, and 
%% $p,q$ for both sensors and actuators
%%(generically called \emph{physical devices}). 
% \marginpar{MM: Sebbene scriviamo questo i nostri valori sono i reali anche
% quando vorremmo booleani, come nell'attuatore $\mathit{cool}!$}
%%%\marginpar{MM: Probably $\cal V = \mathbb{R}$}
\emph{Actuator names} are metavariables for actuator devices like
$\mathit{valve}$, $\mathit{light}$, etc. Similarly, \emph{sensor names}
are metavariables for sensor devices, e.g., a sensor
$\mathit{thermometer}$ that measures, with a given precision, a state
variable called $\mathit{temperature}$.
% are metavariables for sensor devices. For instance, we might have a sensor
% $\mathit{thermometer}$ that measures, with a given precision, a state
% variable called $\mathit{temperature}$.
\emph{Values}, ranged
over by $v,v' \in \cal V$, are built from basic values, such as
Booleans, integers and real numbers; they also include names.

Given a generic set of names $\cal N $, we write $\mathbb{R}^{\cal N} $ to
denote the set of functions %%$[\cal N \rightarrow \mathbb{R} ]$
 assigning a
real value to each name in $\cal N$. For $\xi \in \mathbb{R} ^{\cal N}$,
$n \in \cal N$ and $v \in \mathbb{R} $, we write $\xi [n \mapsto v]$ to
denote the function $\psi \in \mathbb{R} ^{\cal N}$ such that
% $\psi(m)=\xi(m)$, for any $m \neq n$, and $\psi(n)=v$. We say that $\xi
% \in \mathbb{R} ^{\cal M}$ and $\psi \in \mathbb{R} ^{\cal N}$ are
% \emph{disjoint} if ${\cal M} \cap {\cal N} = \emptyset$. When $\xi \in
% \mathbb{R} ^{\cal M}$ and $\psi \in \mathbb{R} ^{\cal N}$ are disjoint, we
% write $\xi \uplus \psi$ to denote the function $\eta \in \mathbb{R}^{{\cal
% M}\cup{\cal N}}$ such that $\eta(m)=\xi(m)$, for any $m\in {\cal M}$, and
% $\eta(n)=\psi(n)$, for any $n\in {\cal N}$.
$\psi(m)=\xi(m)$, for any $m \neq n$, and $\psi(n)=v$.
%%%\mm{We say that $\xi \in \mathbb{R} ^{\cal M}$ and $\psi \in \mathbb{R} ^{\cal N}$ are disjoint if  ${\cal M} \cap {\cal N} = \emptyset$. When $\xi \in \mathbb{R} ^{\cal M}$ and $\psi \in \mathbb{R} ^{\cal N}$ are disjoint
%% we write $\xi \uplus \psi$ to 
%%denote the function $\eta \in  \mathbb{R}^{{\cal M}\cup{\cal N}}$ such 
%%that $\eta(m)=\xi(m)$, for any $m\in {\cal M}$, and $\eta(n)=\psi(n)$, 
%%for any $n\in {\cal N}$. 
%%} 
For $\xi,\xi' \in \mathbb{R}^{\cal N}$, we write $\xi \leq \xi'$ if $\xi(x) \leq
\xi'(x)$, for any $x \in {\cal N}$.
Given   $\xi_1 \in \mathbb{R}^{{\cal N}_1} $ and  $\xi_2 \in \mathbb{R}^{{\cal N}_2} $
such that ${{\cal N}_1} \cap {{\cal N}_2}=\emptyset$,  we denote
with $\xi_1 \uplus \xi_2$ the function in
$\mathbb{R}^{{\cal N}_1 \cup {\cal N}_2} $
such that
$(\xi_1 \uplus \xi_2) (x)=\xi_1(x)$, if $x \in {{\cal N}_1} $, and 
$(\xi_1 \uplus \xi_2) (x)=\xi_2(x)$, if $x \in {{\cal N}_2} $.
Finally, given   $\xi  \in \mathbb{R}^{{\cal N} } $ and   a set of names
 ${\cal M} \subseteq {\cal N}$, we write  $\xi\restrict{\cal M}$ for the   restriction of function $\xi$ to
 the set %% of names 
${{\cal M} }$.

%%\begin{definition}[Cyber-physical system] 
In \cname{}, a \emph{cyber-physical system} consists of two components: a \emph{physical environment} $E$ that encloses all physical aspects of a system (state variables, physical devices, evolution law, etc) and a
\emph{cyber component\/}, represented as a concurrent process $P$ that interacts with the physical devices (sensors and actuators) of the system, and can communicate, via channels, with other processes of the same \CPS{} or with processes of other \CPS{s}. 

 We write $\confCPS E P$ to denote the resulting \CPS, and use 
$M$ and $N$ to range over \CPS{s}. 
%% MASSIMO: ok :)
%%\end{definition} 
%
\enlargethispage{.3\baselineskip}
Let us  formally define physical environments.
%% $E$ and (logical) processes $P$
%%in order to formalise our calculus. 
%%
%% proposal for modelling (and reasoning about)
%%%\CPS{s} and their components.
%
% Let us now  provide more details on both physical environments and concurrent processes in order to describe our proposal for modeling (and reasoning about) \CPS{s} and their components.
%
\begin{definition}[Physical environment]
\label{def:physical-env}
Let $\hat{\mathcal{X}} \subseteq \mathcal{X}$ be a set of state variables,
$\hat{\mathcal{A}} \subseteq \mathcal{A}$ be a set of actuators, and
$\hat{\mathcal{S}} \subseteq \mathcal{S}$ be a set of sensors. A
\emph{physical environment} $E$ is 7-tuple
$\envCPS 
{\statefun{}} 
{\actuatorfun{}} 
{\uncertaintyfun{}}  
{\evolmap{}}
{\errorfun{}}  
{\measmap{}}   
{\invariantfun{}}
$,
where:
\begin{itemize}[noitemsep]
\item $\statefun{} \in \mathbb{R}^{\hat{\cal X}} $ is the
\emph{state function},
\item $\actuatorfun{} \in \mathbb{R}^{\hat{\cal A}} $ is the
\emph{actuator function},
\item $\uncertaintyfun{} \in \mathbb{R}^{\hat{\cal X}} $ is the
\emph{uncertainty function},
\item $\evolmap{}: \mathbb{R}^{\hat{\cal X}} \times
\mathbb{R}^{\hat{\cal A}} \times \mathbb{R}^{\hat{\cal X}} \rightarrow
2^{\mathbb{R}^{\hat{\cal X}} }$ is the \emph{evolution map}, 
\item $\errorfun{} \in \mathbb{R}^{\hat{\cal S}}$ is the
\emph{sensor-error function},
%%and $|\hat{\cal S}| 
%\, \subseteq \, |\hat{\cal X}|$
%% which returns the \emph{accuracy} in measures for each sensor in $\cal S$
\item $\measmap{}: \mathbb{R}^{\hat{\cal X}} \times
\mathbb{R}^{\hat{\cal S}} \rightarrow 2^{\mathbb{R}^{\hat{\cal S}} }$ is
the \emph{measurement map}, 
\item $\invariantfun{}: \mathbb{R} ^{\hat{\cal X}}  
%%%\times \mathbb{R} ^{\cal A}  
\rightarrow \{\true, \false \}$ is the \emph{invariant function}.
% \hfill $\Box$
\end{itemize}
All the functions defining an environment are \emph{total functions\/}.
\end{definition}
 The
\emph{state function} $\statefun{}$ returns the current value (in
$\mathbb{R}$) associated to each state variable of the system. The
\emph{actuator function} $\actuatorfun{}$ returns the current value
associated to each actuator. The \emph{uncertainty function}
$\uncertaintyfun{}$ returns the uncertainty associated to each state
variable. Thus, given a state variable $x \in \hat{\cal X}$,
$\uncertaintyfun{}(x)$ returns the maximum distance between the real value
of $x$ and its representation in the model. 
%%Later in the paper, we will be
%%interested in comparing the accuracy of two systems. Thus, e.g., for
%%$\uncertaintyfun{}, \uncertaintyfun'{} \in \mathbb{R}^{\hat{\cal X}}$, we
%%will write $\uncertaintyfun{} \leq \uncertaintyfun'{}$ if
%%$\uncertaintyfun{}(x) \leq \uncertaintyfun'{}(x)$, for any $x \in
%%\hat{\cal X}$.
 Both the state function and the actuator function are
supposed to change during the evolution of the system, whereas the
uncertainty function is supposed to be constant.

Given a state function, an actuator function, and an uncertainty function,
the \emph{evolution map} $\evolmap{}$ returns the set of next
\emph{admissible state functions}. This function models the
\emph{evolution law} of the physical system, where changes made on
actuators may reflect on state variables. Since we assume an uncertainty in
our models, the evolution map does not return a single state function but
a set of possible state functions. The evolution map is obviously
monotone with respect to uncertainty: if $\uncertaintyfun{} \leq
\uncertaintyfun'{}$ then $\evolmap{}(\statefun{}, \actuatorfun{},
\uncertaintyfun{}) \subseteq \evolmap{}(\statefun{}, \actuatorfun{},
\uncertaintyfun'{})$. Note also that, although the uncertainty function is
constant, it can be used in the evolution map in an arbitrary way (e.g.,
it could have a heavier weight when a state variable reaches extreme
values).

The \emph{sensor-error function} $\errorfun{}$ returns the maximum error
associated to each sensor in $\hat{\cal S}$. Again due to the presence of
the sensor-error function, the \emph{measurement map} $\measmap{}$, given
the current state function, returns a set of admissible measurement
functions rather than a single one.
% Given the current state function and the sensor-error function,
% %%,  and the uncertainty of the next-state function (i.e.\ the evolution model),
% the \emph{measurement map} $\measmap{}$ returns a set of admissible
% % possible
% measurement functions (outputs) rather than a single one.
% % %% for each sensor in $\hat{\mathcal{S}}$.
% % Again, due to the presence of the sensor-error function, the measurement map returns a set of admissible measurement functions rather than a single one.

Finally, the \emph{invariant function} $\invariantfun{}$ represents the
conditions that the state variables must satisfy to allow for the
evolution of the system. A \CPS{} whose state variables don't satisfy the
invariant is in \emph{deadlock}.

Let us now formalise in \cname{} the cyber components of \CPS{s}.
Our (logical) processes build on 
% Hennessy and Regan's timed process algebra TPL~\cite{HR95}
the \emph{timed process algebra TPL}~\cite{HR95} (basically CCS enriched with a discrete notion of time). We extend TPL with  two  constructs: one to read values detected at sensors, and one to write values on actuators. 
The remaining processes of the calculus are the same as those of TPL. 

\begin{definition}[Processes]
\emph{Processes} are defined by the grammar:\\[2pt]
\begin{math}
\begin{array}{rl}	
P,Q \Bdf & \nil \q\, \big| \q \tick.P \q \big| \q P \parallel Q \q 
\big| \q \timeout {\pi.P} {Q} 
\q \big|   \q \ifelse b P Q \q \big| \q P{\setminus} c  
\q \big| \q X \q\big| \q  \fix X P \, .
% \\
% &\hspace*{6.35cm} \Box
\end{array}
\end{math}
% \begin{displaymath}
% \begin{array}{rcl}
% %%M,N & \Bdf & \zero  \Bor   \nodep   n {P}{\mu}{h}  \Bor  M | N   \Bor  \res c M\\[3pt]
% P,Q & \Bdf &  \nil  \Bor  \tick.P    \Bor \timeout {\pi.P} {Q}  \Bor \ifelse b P Q \Bor  \\[3pt]
% &  &   %%\res c P
%     X \Bor  \fix X P \Bor  P \parallel Q  \enspace .
% \end{array}
% \end{displaymath}
\end{definition}

% \marginpar{I think that we should explain all the symbols that we are borrowing ($\Lightning$, ...)}
% \marginpar{LV: ``can transmit ... and then continues'' sounds a bit strange> Maybe ``can transmit ... and in that case continues'', but it is maybe really the ``can'' that could be confusing. Maybe ``if it transmits, then it continues; otherwise''? (And later there is a ``writes'' that should be ``can write''?)}

We write $\nil$ for the \emph{terminated process}. The process $\tick.P$
sleeps for one time unit and then continues as $P$. We write $P \parallel
Q$ to denote the \emph{parallel composition} of concurrent processes 
$P$ and $Q$. The process $\timeout {\pi.P} Q$, with $\pi\in
\{\OUT{c}{v},\LIN{c}{x}, \rsens x s, \wact v a \}$, denotes
\emph{prefixing with timeout}. Thus, $\timeout{\OUT c v . P}Q$ sends the
value $v$ on channel $c$ and, after that, it continues as $P$; otherwise,
if no communication partner is available within one time unit, 
 it evolves into $Q$. The process $\timeout{\LIN c x.
P}Q$ is the obvious counterpart for receiving. $\timeout{\rsens x s.P}{Q}$
reads the value $v$ detected by the sensor $s$ and, after that, it
continues as $P$, where $x$ is replaced by $v$; otherwise, after one time
unit, it evolves into $Q$. $\timeout{\wact v a.P}{Q}$ writes the value $v$
on the actuator $a$ and, after that, it continues as $P$; otherwise, after
one time unit, it evolves into $Q$.
% \mm{
%%\begin{remark}
%%\label{rem:syntactic-attack}
%%\end{remark}
The process $P{\setminus}c$ is the channel restriction operator of CCS. It is quantified over the set $\cal C$ of communication channels but we often use the shorthand 
$P{\setminus}C$ to mean $P{\setminus}{c_1}{\setminus}{c_2}\cdots{\setminus}{c_n}$, for $C= \{ c_1, c_2, \ldots , c_n \}$. 
The process $\ifelse b P Q$ is the standard conditional, where $b$ is a decidable guard. For simiplicity, as in CCS, we identify process $\ifelse b P Q$ with $P$, if $b$ evaluates to true,  and $\ifelse b P Q$  with $Q$, if $b$ evaluates to false.
%%The operator $\res c P$ makes 
%%the channel $c$ private to process $P$.
%%%Rules \rulename{Outp} and \rulename{Inpp}  model transmission and reception along a channel $c$.
In processes of the form $\tick.Q$ and $\timeout {\pi.P} Q$, the occurrence of $Q$ is said to be \emph{time-guarded}. The process $\fix X P$ denotes \emph{time-guarded recursion} as all occurrences of the process variable $X$ may only occur time-guarded in $P$.

In %%processes
the two constructs $\timeout{\LIN c x. P}Q$ and $\timeout{ \rsens x s. P}Q$,
 the variable $x$ is said to be
\emph{bound\/}. Similarly, the process variable $X$ is bound in $\fix X P$.
 %%In the term $\res c P$ the channel $c$ is bound. 
This gives rise to the standard notions of \emph{free/bound (process) variables} %%, \emph{free/bound channels},
and \emph{$\alpha$-conversion}.
We identify processes  up to $\alpha$-conversion (similarly, we identify \CPS{s} up to renaming 
of state variables, sensor names, and actuator names).
A term is \emph{closed} if it does not contain free (process) variables, and %%, although it may contain free channels. 
we assume to always work with closed processes: the absence of free variables is
preserved at run-time. As further notation, we write $T{\subst v x}$ for the substitution of
the variable $x$ with the value $v$ in any expression $T$ of our language.
Similarly, $T{\subst P X}$ is the substitution of the process variable $X$
with the process $P$ in $T$. 
%%We also write $T{\subst {v'} v}$ for the
%%substitution of the  value $v$ with $v'$ in $T$, 
%%if $v$ and $v'$ are compatible values (for instance, they are both actuator names, or sensor names).\marginpar{MM: funziona?}

%%\enlargethispage{2\baselineskip}
The syntax of our \CPS{s} is slightly too permissive as a process might use
sensors and/or actuators  which are not defined in the physical environment. 
%%We could rule out ill-formed \CPS{s}
%%by means of a simple type system. For %the sake of 
%%simplicity, we prefer to provide the following definition.
\begin{definition}[Well-formedness]
\label{def:well-formedness}
Given a process $P$ and an environment $E= \envCPS 
{\statefun{}} 
{\actuatorfun{}} 
{\uncertaintyfun{}}  
{\evolmap{}}
{\errorfun{}}  
{\measmap{}}   
{\invariantfun{}}
$, the \CPS{} $\confCPS E P$ is %said to be 
\emph{well-formed} if: (i) for any sensor $s$ mentioned in $P$, the function $\errorfun{}$ is defined in $s$; (ii) for any actuator $a$ mentioned in $P$, the function $\actuatorfun{}$ is defined in $a$. 
%%% (iii) $P$ doesn't contain prefixes of the form $\wact v s$ or $ \rsens x a$. 
% \hfill $\Box$
\end{definition}
Hereafter, we will always work with well-formed networks.

Finally, we assume a number of \emph{notational conventions}.  
%%$\prod_{i \in I}P_i$ denotes the parallel composition of all $P_i$, for $i{\in}I$, where $\prod_{i \in I}P_i = \nil$ if $I =\emptyset$. We write $\prod_{i}P_i$ when 
%%$I$ is not relevant. 
%%We write $\timeout{\pi}{\nil}$ instead of $\timeout{\pi.\nil}{\nil}$.
%%To model \emph{time-persistent prefixing}, 
We write $\pi.P$ instead of $\fix{X}\timeout{\pi.P}X$, when $X$ does not occur in $P$. 
%%We write $\timeout{\pi}Q$ as an abbreviation for $\timeout{\pi.\nil}{Q}$ and $\timeout{\pi.P}{}$ instead of $\timeout{\pi.P}{\nil}$. 
We write $\OUTCCS c$ (resp.\ $\LINCCS c$) when channel  $c$ is used for pure synchronisation.
For $k\geq 0$, we write $\tick^{k}.P$ as a shorthand for $\tick.\tick. \ldots \tick.P$, where the prefix $\tick$ appears $k$ consecutive times. 
%%We write $\ifthen b P$ instead of $\ifelse b P \nil$. 
%%We use $\res {\tilde c} P$ as an abbreviation for $(\nu {c_1})\ldots (\nu {c_k})P$, 
%%with $\tilde{c}=c_1, \ldots , c_k$. 
Given $M = \confCPS E  P$, we write 
$M\parallel Q$ for $\confCPS E  { (P\parallel Q) }$, and 
$M {\setminus} c$ for $\confCPS E {P{\setminus}c}$.

%%%%%%%%%%%%%%%%%%%%%%%%%%%%%%%%%%%%%%%%%%%%%%%%%%%%%%%%%%%%%%%%%%%
\subsection{Labelled Transition Semantics}
\label{lab_sem}

%\enlargethispage{1.5\baselineskip}
%%The physical evolution of a plant 
%%is given via the evolution map reported in $E$. 
In this section, we provide the dynamics of \cname{}  
 in terms of a \emph{labelled
transition system (LTS)} in the SOS style of Plotkin. %~\cite{Plo04}.
% of Plotkin~\cite{Plo04}. 
In \autoref{def:op_env}, for convenience, we define some auxiliary
operators on environments.
\begin{definition}
\label{def:op_env}
Let $E = \envCPS {\statefun{}} 
{\actuatorfun{}} 
{\uncertaintyfun{}}  
{\evolmap{}}
{\errorfun{}}  
{\measmap{}}   
{\invariantfun{}}$. % We define: %\\[1pt]
\vspace*{-2.5mm}
\begin{itemize}[noitemsep]
\item $\mathit{read\_sensor}(E,s)  =  \{\xi(s) : \xi \in \measmap{}(\statefun{},\errorfun{})\}$,
\item $\mathit{update\_act}(E,a,v)  = 
\envCPS {\statefun{}} {\actuatorfun{}[a {\mapsto} v]}
{\uncertaintyfun{}} {\evolmap{}} {\errorfun{}} {\measmap{}}  {\allowbreak \invariantfun{}}$,
%%\item $E \subst {\xi} {\uncertaintyfun{}}  \deff
%%\envCPS {\statefun{}} {\actuatorfun{}} {\xi} {\evolmap{}} {\errorfun{}} {\measmap{}} {\invariantfun{}}$,
\item $\mathit{next}(E)  = \bigcup_{\xi \in   \evolmap{}(\statefun{}, \actuatorfun{}, \uncertaintyfun{})} 
\{ \envCPS {\xi} 
{\actuatorfun{}} 
{\uncertaintyfun{}}  
{\evolmap{}}
{\errorfun{}}  
{\measmap{}}   
{\invariantfun{}} \}$,
\item $\invariantfun{}(E)  = \invariantfun{}(\statefun{})$.
% \hfill $\Box$
\end{itemize}
% \(   
% \begin{array}{l}
% - \, \mathit{read\_sensor}(E,s)  =   \{\xi(s) : \xi \in \measmap{}(\statefun{},\errorfun{})\}\\[2pt]
% - \, \mathit{update\_act}(E,a,v)  =
% \envCPS {\statefun{}} {\actuatorfun{}[a \mapsto v]}
% {\uncertaintyfun{}} {\evolmap{}} {\errorfun{}} {\measmap{}}  {\invariantfun{}}\\[3pt]
% -\, E \subst {\uncertaintyfun'{}} {\uncertaintyfun{}}  =
% \envCPS {\statefun{}} {\actuatorfun{}} {\uncertaintyfun'{}} {\evolmap{}} {\errorfun{}} {\measmap{}} {\invariantfun{}}\\[2pt]
%
% % \operatorname{change\_uncert}(E,\gamma)& \deff &
% %\envCPS {\xi_{\operatorname{x}}} {\xi_{\operatorname{a}}} {\gamma} {\operatorname{next}} {\epsilon}   {\operatorname{sens}}  {\invariantfun{}}\\[5pt]
%
% -\, \mathit{next}(E)  =  \bigcup_{\statefun'{} \in   \evolmap{}(\statefun{}, \actuatorfun{}, \uncertaintyfun{})}
% \envCPS {\statefun'{}}
% {\actuatorfun{}}
% {\uncertaintyfun{}}
% {\evolmap{}}
% {\errorfun{}}
% {\measmap{}}
% {\invariantfun{}}\\[2pt]
%
% - \, \invariantfun{}(E)  = \invariantfun{}(\statefun{}) \enspace . \\
% \hspace*{8cm} \Box
% \end{array}
% \)
\end{definition}
 The operator 
$\mathit{read\_sensor}(E,s)$ returns the set of possible measurements detected by sensor $s$ in the environment $E$; it returns a set of possible values rather than a single value due to the error $\errorfun{}(s)$ of sensor $s$. 
$\mathit{update\_act}(E,a,v)$ returns the new environment in which the
actuator function is updated in such a manner to associate the actuator
$a$ with the value $v$.
%%$E \subst {\xi} {\uncertaintyfun{}}$
% $\operatorname{change\_uncert}(E,\gamma)$
%%returns a new environment in which the uncertainty function
%%$\uncertaintyfun{}$ is replaced by $\xi$. Given a \CPS{} $M
%%= \confCPS E P$, we often write $M \subst {\uncertaintyfun'{}}
%%{\uncertaintyfun{}}$ for $\confCPS {E \subst {\uncertaintyfun'{}}
%%{\uncertaintyfun{}}} P$.
 $\mathit{next}(E)$ returns the set of the next
admissible environments reachable from $E$, by an application of the
evolution map. $\invariantfun{}(E)$ checks whether the state variables
satisfy the invariant  (here, with an
abuse of notation, we overload the meaning of the function
$\invariantfun{}$).

\begin{table}[t]
\begin{displaymath}
\begin{array}{l@{\hspace*{3mm}}l}
\Txiom{Outp}
{-}
{ { \timeout{\OUT c v .P}Q } \trans{\out c v}   P}
&
\Txiom{Inpp}
{-}
{ { \timeout{\LIN c x .P}Q } \trans{\inp c v}    {P{\subst v x}}  }

\\[14pt]
\Txiom{Write}
{ - }  %%%p \in \{ s, a \} \Q p \textrm{ data-driven}}
{ { \timeout{\wact v a .P}Q } \trans{\snda a v}   P}
&
\Txiom{Read}
{  - } %%%%p \in \{ s , a \} \Q p \textrm{ data-driven}}
{ { \timeout{\rsens x s .P}Q } \trans{\rcva s v}    {P{\subst v x}}  }
\\[14pt]
\Txiom{Com}
{ P \trans{\out c v}  { P'}  \Q  Q \trans{\inp c v}  { Q'} }
{ P \parallel  Q \trans{\tau}  {P'\parallel Q'}}
&
\Txiom{Par}
{ P \trans{\lambda}  P' \Q \lambda \neq  \tick }
{ {P\parallel Q} \trans{\lambda} {P'\parallel Q}}
\\[14pt]
\Txiom{ChnRes}{P \trans{\lambda} P' \Q \lambda \not\in \{ {\inp c v}, {\out c v} \}}{P {\setminus}c \trans{\lambda} {P'}{\setminus}c}
&
\Txiom{Rec}
{  {P{\subst {\fix{X}P} X}} \trans{\lambda}  Q}
{ {\fix{X}P}  \trans{\lambda}  Q}
\\[14pt]
\Txiom{TimeNil}{-}
{ \nil \trans{\tick}  \nil}
& 
\Txiom{Delay}
{-}
{  { \tick.P} \trans{\tick}  P}
\\[14pt]
\Txiom{Timeout}
{-}
{  {\timeout{\pi.P}{Q} }   \trans{\tick}  Q}
&
\Txiom{TimePar}
{
  P \trans{\tick}  {P'}  \q\,
   Q \trans{\tick} {Q'}  \q\, P \parallel Q \not\!\!\!\trans{\tau}
}
{
  {P \parallel Q}   \trans{\tick}  { P' \parallel Q'}
}
%%\Txiom{Res}
%%{ P \trans{\alpha}  {P'} \Q  \lambda \not \in  \{\out c v , \inp c v  \} }
%%{ {\res c P} \trans{\lambda} {\res c {P'}} }
\end{array}
\end{displaymath}
\caption{LTS for processes}
\label{tab:lts_processes} 
\end{table}

%\enlargethispage{\baselineskip}
\begin{table}[t]
\begin{displaymath}
\begin{array}{c}
\Txiom{Out}
{P \trans{\out c v}  P' \Q \invariantfun{}(E)}
{\confCPS E  P   \trans{\out c v}   \confCPS E  {P' }}
\Q\Q\Q\Q
\Txiom{Inp}
{P  \trans{\inp c v}  P'\Q \invariantfun{}(E)}
{\confCPS E  P    \trans{\inp c v}  \confCPS E {P' }}
\\[16pt]
\Txiom{SensRead}{P \trans{\rcva s v} P'  \Q \invariantfun{}(E)\Q
\mbox{\small{$v \in \mathit{read\_sensor}(E,s)$}} 
}
{\confCPS E  P \trans{\tau} \confCPS E {P'}}
\\[16pt]
\Txiom{ActWrite}{P \trans{\snda a v} {P'}  \Q   \invariantfun{}(E) \Q
{E'}=\operatorname{update\_act}(E,a,v)}
{\confCPS E  P \trans{\tau} \confCPS {E'}{P'}}
\\[16pt]
\Txiom{Tau}{P \trans{\tau} P' \Q \invariantfun{}(E)}
{ \confCPS E  P \trans{\tau} \confCPS E  {P'}}
\Q\;
\Txiom{Time}{ P \trans{\tick} {P'} \q\;
\confCPS E  P \ntrans{\tau} \q\;
\invariantfun{}(E) \q\; E' \in \operatorname{next}(E) }
{\confCPS E  P \trans{\tick} \confCPS {E'}  {P'}}
\end{array}
\end{displaymath}
\caption{LTS for \CPS{s}}
\label{tab:lts_systems} 
\end{table}
In \autoref{tab:lts_processes}, we provide transition rules for processes.
Here, the meta-variable $\lambda$ ranges over labels in the set 
 $\{\tick,
\tau, {\out c v}, {\inp c v}, \allowbreak \snda a v,\rcva s v \}$. Rules
\rulename{Outp}, \rulename{Inpp} and \rulename{Com} serve to model channel
communication, on some channel $c$. Rules~\rulename{Write} denotes the 
writing of some data $v$ on an actuator $a$. Rule~\rulename{Read} denotes the reading of some data $v$ via a sensor $s$. Rule \rulename{Par} propagates untimed actions over parallel components.
Rules  \rulename{ChnRes} and \rulename{Rec} are the standard rules for 
channel restriction and recursion, respectively. The following four rules 
are standard, and model the passage of one time unit.  The symmetric counterparts of rules \rulename{Com} 
%%,
%%\rulename{$\mbox{\Lightning}$SensWrite$\mbox{\,\Lightning}$},
%%\rulename{$\mbox{\Lightning}$ActRead$\mbox{\,\Lightning}$} 
and \rulename{Par} are obvious and thus omitted from the table.\\
\indent 
In \autoref{tab:lts_systems}, we lift the transition rules from processes
to systems. All rules have a common premise $\invariantfun{}(E)$: 
a \CPS{} can evolve only if the invariant is
satisfied, otherwise it is deadlocked. Here, actions, ranged over by $\alpha$, are in the set $\{\tau,
{\out c v}, {\inp c v}, \tick \}$. These actions 
denote: non-observable activities ($\tau$); 
observable logical activities, \emph{i.e.\/},
channel transmission (${\out c v}$ and ${\inp c v}$); the passage
of  time ($\tick$). 
%%\enlargethispage{\baselineskip}
%%(The choice of this set will
%%become clear once we describe the LTS.) 
Rules \rulename{Out} and
\rulename{Inp} model transmission and reception, with an external system,
on a channel $c$. Rule \rulename{SensRead} models the reading of the
current data detected at sensor $s$. 
 Rule~\rulename{ActWrite} models the writing of a value $v$ on an
 actuator $a$. 
Rule \rulename{Tau} lifts non-observable actions from processes to
systems.  A similar lifting occurs in rule
\rulename{Time} for timed actions, where $\operatorname{next}(E)$ returns
the set of possible environments for the next time slot. Thus, by an
application of rule \rulename{Time} a \CPS{} moves to the next physical
state, in the next time slot. \\
\indent 
Now, having defined the actions that can be performed by a \CPS{}, we can
easily concatenate these actions to define \emph{execution traces\/}.
 Formally, given a trace  $t = \alpha_1 \ldots
\alpha_n$, we will write $\trans{t}$ as an abbreviation for
$\trans{\alpha_1}\ldots \trans{\alpha_n}$. \\
%%In the rest of the paper, we
%%will use the function $\#\tick(t)$ to get the number of occurrences of the
%%timed action $\tick$ in the trace $t$.
\indent 
Below,  we report a few  desirable time properties which hold in our calculus:
(a) \emph{time determinism\/}, (b) \emph{maximal progress\/}, (c)
\emph{patience\/}, 
 and (d) \emph{well-timedness\/} (symbol $\equiv$ denotes standard \emph{structural congruence}
for timed processes~\cite{Mil91,MBS11}).
\begin{theorem}[Time properties] 
\label{prop:time}	
Let $M = \confCPS E P$. 
%%[Localised time determinism]
\begin{itemize}
\item[(a)]
\label{prop:timed}
If $M \trans \tick \confCPS {\hat{E}} {Q} $ and 
$M\trans \tick \confCPS {\tilde{E}} {R}$, then $\{ \hat{E}, \tilde{E} \} 
\subseteq \operatorname{next}{(E)}$ and 
$Q \equiv  R$.
\item[(b)]
\label{prop:maxprog}
If $M\trans \tau M'$  then there is no $M''$ such that  $M\trans \tick M''$. 
\item[(c)]
\label{prop:patience}
If $M \trans \tick M'$  for no  $M'$ 
 then either  $ \operatorname{next}(E )=\emptyset$ or
 $\invariantfun{(M)}=\false$ or
there is $N$ such that $M \trans \tau N$. 

%%%\item 
%%%\label{prop:patience2}
%%If $M \trans \alpha M'$  for no  $M'$, then $\invariantfun{(M)}=\false$.
%%{\bf la vogliamo?}
\item[(d)]
\label{prop:welltime}
For any $M$  there is a $ k\in\mathbb{N}$ such that  if
$M  \trans{\alpha_1}\dots \trans{\alpha_n} N$, with $\alpha_i \neq \tick$,   then $n\leq k$. 
\end{itemize}
\end{theorem}
 \emph{Well-timedness}~\cite{MBS11,CHM15} ensures the absence of infinite 
instantaneous 
traces  which would prevent the passage of time, and hence the 
physical evolution of a \CPS{}.

%%%%%%%%%%%%%%%%%%%%%%%%%%%%%%%%%%%%%%%%%%%%%%%%%%%%%%%%%%%%%%%%%%%%

\section{Bisimulation}
\label{sec:bisimulation}

Once defined the labelled transition semantics, we are ready to define our 
bisimulation-based behavioural equality for \CPS{s}. We recall that the only 
\emph{observable activities} in \cname{} are: 
time passing and channel communication.  
%%and system deadlock.
 As a consequence, the capability to observe physical events 
%%(different from deadlocks) 
depends on the capability of the cyber components to recognise those events by acting on sensors and actuators, and then signalling
them using (unrestricted) channels.

%% which will be showed to be
%% both sound and complete with respect to our contextual equivalence. 
%In order to provide our notion of bisimulation 
We adopt a standard notation for weak transitions: we write $\Trans{}$ for the reflexive and transitive closure of $\tau$-actions, namely $(\trans{\tau})^*$, whereas $\Trans{\alpha}$ means $\ttranst{\alpha}$, and finally $\ttrans{\hat{\alpha}}$ denotes $\Trans{}$ if $\alpha=\tau$ and $\ttrans{\alpha}$ otherwise.
\begin{definition}[Bisimulation]
\label{def:bisimulation}
A binary symmetric relation $\RR$ over \CPS{s} is a bisimulation if 
 $M \RRr N$ and $M \trans{\alpha}
M'$ implies that there exists $N'$ such that $N\Trans{\hat{\alpha}}N'$
and $M' \RRr N'$. We say that $M$ and $N$ are bisimilar, written $M \approx N$, 
if $M \RRr N$ for some bisimulation $\RR$. 
\end{definition}

A main result of the paper is that our bisimilarity can be used to 
compare \CPS{s} in a compositional manner. In particular, our bisimilarity is
 preserved   by 
parallel composition of (non-interfering) \CPS{s}, 
by parallel composition of (non-interfering) processes,
 and by channel restriction. 

Two \CPS{s} do not interfere with each other 
if they have a disjoint physical plant. Thus, let   
$ E^i = \envCPS
{\statefun^i{} } 
{\actuatorfun^i{} } 
{\uncertaintyfun^i{} }  
{\evolmap^i{} }
{\errorfun^i{} }  
{\measmap^i{} }   
{\invariantfun^i{} }
$
with sensors in   $\hat{\mathcal{S}}_i$,  actuators in  $\hat{\mathcal{A}}_i$, 
  and state variables in 
  $\hat{\mathcal{X}}_i$,  for $i \in \{ 1,2\}$.
If $\hat{\mathcal{S}}_1 \cap \hat{\mathcal{S}}_2=\emptyset$ and 
$\hat{\mathcal{A}}_1 \cap \hat{\mathcal{A}}_2=\emptyset$ and
$\hat{\mathcal{X}}_1 \cap \hat{\mathcal{X}}_2=\emptyset$, then
we define  the \emph{disjoint union} of the environments $E_1$ and $E_2$,
written $E_1 \uplus E_2$,  to be the environment 
$   \envCPS
{\statefun{} } 
{\actuatorfun{} } 
{\uncertaintyfun{} }  
{\evolmap{} }
{\errorfun{} }  
{\measmap{} }   
{\invariantfun{} }
$
such that: 
${\statefun{} }=  \statefun^1{} \uplus \statefun^2{}  $, 
${\actuatorfun{} }=\actuatorfun ^1{} \uplus \actuatorfun^2{}  $, 
${\uncertaintyfun{} }=\uncertaintyfun^1{} \uplus \uncertaintyfun^2{}  $, 
${\errorfun{} } = \errorfun^1{} \uplus \errorfun^2{}  $, and 
\begin{center}
\begin{math}
\begin{array}{rcl}
{\evolmap{} }(\xi,\psi , \phi) & = & \{ \xi'= \xi_1 \uplus \xi_2 \,:\,
\xi_i \in {\evolmap^i{} }(\xi\restrict{\hat{\mathcal{X}}_i},\psi\restrict{\hat{\mathcal{A}}_i} , 
\phi\restrict{\hat{\mathcal{X}}_i}), \mbox{ for }
i \in \{ 1, 2 \} \} \\
{\measmap{} } (\xi,\psi)& =  & \{ \xi' =   \xi_1 \uplus \xi_2 \,:\,
\xi_i \in
{\measmap^i{} }(\xi\restrict{\hat{\mathcal{X}}_i},\psi\restrict{\hat{\mathcal{S}}_i} ), \mbox{ for }
i \in \{ 1,2 \} \}  \\
 \invariantfun{} (\xi ) & = &\invariantfun^1{} (\xi\restrict{\hat{\mathcal{X}}_1 } )\wedge \invariantfun^2{} (\xi\restrict{\hat{\mathcal{X}}_2} ) \enspace . 
\end{array}
\end{math}
\end{center}

\begin{definition}[Non-interfering {\CPS}s]
Let $M_i = \confCPS {E_i}{P_i}$, for $i \in \{ 1 , 2 \}$. We say that 
$M_1$ and $M_2$ do not interfere with each other if
$E_1$ and $E_2$ have disjoint sets of state variables, sensors and actuators. In this case,  we write $M_1 \uplus M_2$ to denote the \CPS{}  defined as $\confCPS {(E_1\uplus E_2)}{(P_1\parallel P_2)}$.
\end{definition}

A similar but simpler definition can be given for processes.  
Let $M=\confCPS E P$, 
 a non-interfering process $Q$ is a process which does 
not interfere with the plant $E$ as it never accesses its
sensors and/or actuators. Thus, in the system $M \parallel Q$
the process $Q$ cannot interfere with the physical evolution of $M$. 
However, process $Q$ can definitely affect the observable behaviour
of the whole system by communicating on channels. Notice that, 
as  we only consider well-formed \CPS{s} (\autoref{def:well-formedness}), a non-interfering 
processes is basically a (pure) TPL process~\cite{HR95}. 
\begin{definition}[Non-interfering processes]
A process $P$ is called \emph{non-interfering} if it never acts on 
sensors and/or actuators.
\end{definition}

Now, everything is in place to prove the compositionality 
of our bisimilarity $\approx$. 
\begin{theorem}[Congruence results]
Let $M$ and $N$ be two \CPS{s}. 
\label{thm:congruence}
 \begin{enumerate}
%%\item  $   P \, \simeq \,  Q$ implies that $\confCPS E P \simeq \confCPS {E } {Q}$, for any environment $E$. 
\item
\label{thm:congruence1}
 $M \approx N$ implies $M \uplus O \approx N \uplus O$, for any 
non-interfering \CPS{} $O$ 
\item
\label{thm:congruence2}
 $M \approx N$ implies $M \parallel P \approx N \parallel P$, for any 
non-interfering process $P$ 
%%\item 
%% $\confCPS E P \simeq \confCPS {E' } {P'}$ implies $ \confCPS {(E \uplus E'') } {(P\parallel P'')} \simeq  \confCPS  {(E' \uplus E'' )} {(P'\parallel P'')}$  for any well formed $\confCPS  { E'' } {  P''} $.
\item 
\label{thm:congruence3}
$M \approx N$ implies $M {\setminus} c \; \approx \; M {\setminus} c$, for 
any channel $c$.

%%\item 
%%\label{thm:congruence4}
%%$M \approx N$ implies $M\rho \approx N \rho$, for any renaming $\rho$.
%%%\marginpar{MM: Definire per bene il renaming.}
\end{enumerate} 
\end{theorem}
The presence of invariants in the definition of physical environment
makes the proof of the second item of the theorem above
%%that $\approx$
%%is preserved by the operator $\uplus$ 
non standard.  

As we will see in the next section, these compositional properties will 
be very useful when reasoning about complex systems.

\section{Case study}
\label{sec:case-study}
In this section, we 
%%provide a case study  to illustrate
%%how \cname{} can be used to specify and reason on \CPS{s} in a compositional 
%%manner. 
%%We remark that while we have kept the example as simple as possibile, 
%%it is actually far from trivial, 
%% A more complex example (say, with $n$
%%sensors and $m$ actuators) wouldn't have been more instructive but just
%%made the paper more dense.
%%
%%Suppose to 
model  in \cname{} an engine, called $\mathit{Eng}$, whose temperature is
maintained within a specific range by means of a cooling system. The
physical environment $\env$ of the engine 
  is constituted by: (i) a state variable
$\mathit{temp}$ containing the current temperature of the engine; (ii) an
actuator $\mathit{cool}$ to turn on/off the cooling system; (iii) a sensor
$s_{\mathrm{t}}$ (such as a thermometer or a thermocouple) measuring the temperature
of the engine; (iv) an uncertainty $\delta=0.4$ associated to the only
variable $\mathit{temp}$; (v) a simple evolution law that increases
(resp., decreases) the value of $\mathit{temp}$ of one degree per time
unit if the cooling system is inactive (resp., active) --- the evolution
law is obviously affected by the uncertainty $\delta$; (vi) an error
$\epsilon =0.1$ associated to the only sensor $s_{\mathrm{t}}$; (vii) a measurement
map to get the values detected by sensor $s_{\mathrm{t}}$, up to its error
$\epsilon$; (viii) an invariant function saying that the system gets
faulty when the temperature of the engine gets out of the range $[0, 30]$.

%%However, well before that, when the measured temperatures
%%is above $10$ degrees, the logics of the system uses the
%% channel  $alarm$  to signal the protracted overheating 
%%of the engine.

Formally,  $\env = \envCPS 
{\statefun{}} 
{\actuatorfun{}} 
{\uncertaintyfun{}}  
{\evolmap{}}
{\errorfun{}}  
{\measmap{}}   
{\invariantfun{}}$ with:
\begin{itemize}[noitemsep]
\item $\statefun{} \in \mathbb{R} ^{\{\mathit{temp}\} }$ and 
$\statefun{}(\mathit{temp})=0$;
\item $\actuatorfun{} \in \mathbb{R} ^{\{\mathit{cool}\} }$ and
$\actuatorfun{}(\mathit{cool})=\off$; for the sake of simplicity, we can
assume $\actuatorfun{}$ to be a mapping $\{ \mathit{cool} \} \rightarrow
\{ \on , \off\}$ such that $\actuatorfun{}(\mathit{cool})= \off$ if
$\actuatorfun{}(\mathit{cool}) \geq 0$, and $\actuatorfun{}(\mathit{cool})= \on$ if
$\actuatorfun{}(\mathit{cool}) < 0$;

\item $\uncertaintyfun{} \in \mathbb{R} ^{\{\mathit{temp}\} }$ and
$\uncertaintyfun{}(\mathit{temp})=0.4=\delta$;

\item $\evolmap{}( \statefun^i{}, \actuatorfun^i{}, \uncertaintyfun{}) =
\big \{ \, \xi : \xi(\mathit{temp}) = \statefun^i{}(\mathit{temp}) + 
\mathit{\mathit{heat}}(\actuatorfun^i{},\allowbreak \mathit{cool}) + \gamma \;\,
\wedge \;\, \gamma \in [- \delta, + \delta]  \, \big \}$,
%%for any $  \xi^i_{\operatorname{x}}\in \mathbb{R} ^{\{temp\} }$, $   \xi^i_{\operatorname{a}} \in \mathbb{R} ^{\{cool\} }$ and $   \delta\in \mathbb{R} $; 
%%moreover, 
%as $\operatorname{active}(\xi^i_{\operatorname{a}}, cool) \in \{ -1,+1 \}$,  when the heat system is inactive (resp., active) the temperature increases (resp., decreases) of some value $v$, $v \in [1- \delta,1+ \delta]$; 
%%formally, 
%%for any $\xi \in \mathbb{R}^{ \{ cool \} }$, 
where $\mathit{heat}(\actuatorfun^i{},\mathit{cool})=-1$ if
$\actuatorfun^i{}(\mathit{cool}) = \on$ (active cooling), and
$\mathit{heat}(\actuatorfun^i{},\mathit{cool})=+1$ if
$\actuatorfun^i{}(\mathit{cool}) = \off$ (inactive cooling);

\item $\errorfun{} \in \mathbb{R}^{\{s_{\mathrm{t}} \}}$ and
$\errorfun(s_{\mathrm{t}})= 0.1=\epsilon$;

\item $\measmap{}(\statefun^i{}, \errorfun{}) = \big \{ \xi :
\xi(s_{\mathrm{t}}) \in [ \statefun^i{}(\mathit{temp}){-}\epsilon \, , \, \allowbreak
\statefun^i{}(\mathit{temp}){+} \epsilon ]  \big \}$;

%% for any $  \xi^i_{\operatorname{s}}\in \mathbb{R} ^{\{temp\} }$ and 
%% $   \epsilon \in \mathbb{R} ^{\{s_{\operatorname{t}}\} }$; 
\item $\invariantfun{}(\statefun{})=\true$ if $0 \leq \statefun{}(\mathit{temp})\leq 30$; \allowbreak $\invariantfun{}(\statefun{})=\false$, otherwise. 
%%for any $  \xi_{\operatorname{x}}\in \mathbb{R} ^{\{temp\} }$ and $   \xi_{\operatorname{a}} \in \mathbb{R} ^{\{cool\} }$;  
 %Thus, the only admissible values of actuator
%%$\mathit{cool}$ are in $\{-1,+1\}$ and  the engine becomes faulty (goes in deadlock) when the temperature exceeds $50$ degrees. 
%%\marginpar{RL: Se usassimo   un canale ok quando il refrigerante ha fatto il suo dovere allora un attacco ammissibile consiste mell' aumentare le volte che si accende il refrigerante senza far scattare l'allarme. Dipende da cosa vogliamo sia visibile ma se rendiamo visibile troppe cose il sistema non tollera niente. }
% \hfill $\Box$
\end{itemize}
%%\end{example}

%%The cyber component of  $\mathit{Eng}$ consists of a process $\mathit{Ctrl}$
%%which activates the cooling system whenever the (sensed) temperature is 
%%above the threshold $10$. 
%% Intuitively, $\mathit{Ctrl}$ senses the temperature of the engine at each time slot. When the sensed temperature is above $10$ degrees, the controller
%% activates the coolant. The cooling activity is maintained for $5$ consecutive time units. After that time,
%%the temperature is checked again. If the temperature is not above 
%%$10$ than the controller exits the warning state, otherwise 
%%it moves into an \emph{alarm state} and keeps cooling until the temperature
%%drops below $10$.

The cyber component of $\mathit{Eng}$ consists of a process $\mathit{Ctrl}$ 
which models the controller activity. Intuitively, process $\mathit{Ctrl}$ senses the temperature of the engine at each time interval. When the sensed temperature is above $10$, the controller activates the coolant. The cooling activity is maintained for $5$ consecutive time units. After that time, if the temperature does not drop below $10$ then the controller transmits its $\mathit{ID}$ on a specific channel for signalling a  $\mathit{warning}$, it keeps cooling  for another $5$ time units, and then checks again the sensed temperature; otherwise, if the 
sensed temperature is not above the threshold $10$, the controller turns off the cooling and moves to the next time interval.
 Formally, 
\begin{displaymath}
\begin{array}{rcl}
\mathit{Ctrl} & \; = \; & \fix{X} \rsens x {s_{\operatorname{t}}} . 
\ifelse {x>10}
 { \mathit{Cooling}}
 {  \tick.X } \\[1pt]
\mathit{Cooling} & \;  = \;  &   \wact{\on}{\mathit{cool}}. \fix{Y}
%%\OUT{\mathit{warning}}{\mathsf{on} } .
 \tick^5 .   \rsens x {s_{\operatorname{t}}} . \\
&&
\ifelse {x>10} {\OUT{\mathit{warning}}{\mathrm{ID}}.Y}
{\wact{\off}{\mathit{cool}}.\tick.X } \enspace . 
%%\\[3pt]
%%\mathit{Alarm} & \deff &  \fix{Z}  \OUT{\mathit{alarm}}{\mathsf{on} }. \tick^5 .  \rsens x {s_{\operatorname{t}}} . \\
%%&&
%%\mathsf{if} \, (x>10) \, \{ Z \} \; \mathsf{else}\\
%%&&
%%\q\,\{ \OUT{\mathit{alarm}}{\mathsf{off} } .\mathsf{if} \, (x>10) \, \{ Y \} \; \mathsf{else} \: 
%% \{  \wact{\off}{\mathit{cool}}.  \tick .X \} \}  \enspace . 
 \end{array}
%%%\marginpar{RL su Alarm ho messo il warning off, si puo' spostare su $Alm_i$. Inoltre bisognerebbe mettere i timeouts }
\end{displaymath}
The whole engine is defined as: 
% \marginpar{MM: ho aggiunto $\on$ nell'allarme. Poi vi spiego..}
\begin{math}
\mathit{Eng} \: = \:  \confCPS {\env} { \mathit{Ctrl} } \,,
\end{math}
where $\env$ is the physical environment defined before.

Our operational semantics allows us to formally prove a number of 
\emph{run-time properties\/} of our engine. 
For instance,  the following proposition says that our engine 
never reaches a warning state and  never deadlocks.  
%%For instance,  the following proposition says that our engine 
%%% never deadlocks and 
never reaches a warning state.  
%%%We recall that our operational semantics produces a special action $\dead$ when the invariant of a \CPS{} is violated.
\begin{proposition} 
\label{prop:sys}
Let $\mathit{Eng}$ be the \CPS{} defined before. 
If $\mathit{Eng} \trans{\alpha_1} \ldots \trans{\alpha_n} \mathit{Eng}'$, for some $\mathit{Eng}'$, then
  $\alpha_i \in \{ \tau , \tick  \}$, for $1 \leq i \leq n$, and
 there is $\mathit{Eng}''$ such that $\mathit{Eng}' \trans{\alpha} \mathit{Eng}''$, for some 
$\alpha_i \in \{ \tau , \tick  \}$. 
\end{proposition}

Actually, we can be quite precise on the temperature reached by the engine
 before and after the cooling activity: in each of the $5$ time slots of cooling, the temperature will drop of a value laying in the
interval $[1{-}\delta, 1 {+} \delta]$, where $\delta$ is the uncertainty
of the model. Formally, 

\begin{proposition}
\label{prop:X}
 For any execution trace of $\mathit{Eng}$, we have:
\begin{itemize}[noitemsep]
\item when $\mathit{Eng}$ \emph{turns on} the cooling, the value of
the state variable $\mathit{temp}$ ranges over $(10-\epsilon \, , \,  11+ \epsilon  + \delta]$;
%%\item process $\mathit{Check}$ will always sense a temperature less than or equal $10$;
\item when $\mathit{Eng}$ \emph{turns off} the cooling, the value of
the %state 
variable $\mathit{temp}$ ranges over $(10-\epsilon -5{*}(1{+}\delta) \, , \,  11+ \epsilon + \delta -5{*}(1{-}\delta)]$. 
\end{itemize}
\end{proposition}

\begin{figure}[t]
\centering
%%\includegraphics[width=10.5cm,keepaspectratio=true,angle=0]{./figures/nominalBehavior-eps-converted-to.pdf}
%%\\
\includegraphics[width=5.8cm,keepaspectratio=true,angle=0]{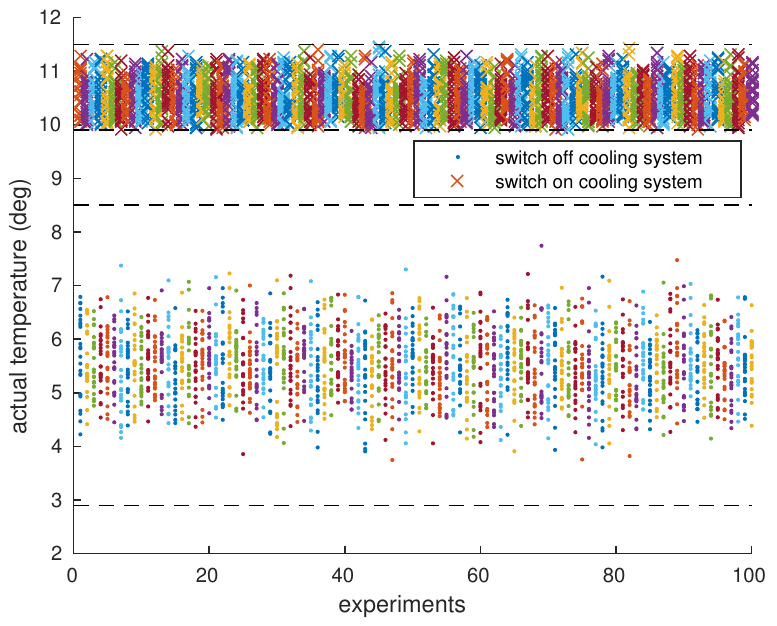}
%%\includegraphics[width=6.4cm,keepaspectratio=true,angle=0]{./figures/CPSnominal_actual-eps-converted-to.pdf}
%\\[8pt]
\Q
\includegraphics[width=5.8cm,keepaspectratio=true,angle=0]{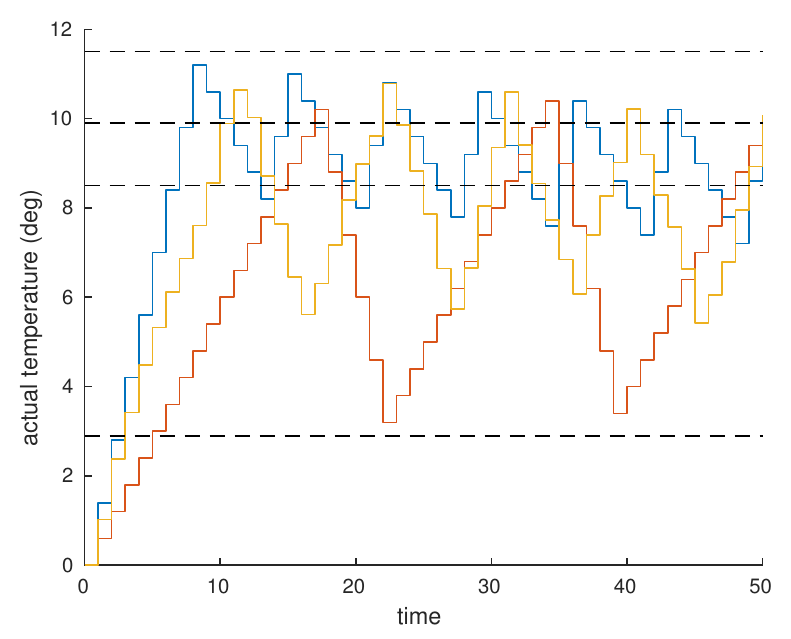}
\caption{Simulations in MATLAB of the engine $\mathit{Eng}$}
\label{f:HS traj}
\end{figure}
%%\begin{figure*}[t]
%%\centering
%%\includegraphics[width=10.5cm,keepaspectratio=true,angle=0]{./figures/nominalBehavior-eps-converted-to.pdf}
%%\\
%%\includegraphics[width=5.75cm,keepaspectratio=true,angle=0]{./figures/CPSnominal_actual_Ex1_crop.pdf}
%%%\includegraphics[width=6.4cm,keepaspectratio=true,angle=0]{./figures/CPSnominal_actual-eps-converted-to.pdf}
%%\\[8pt]
%\Q\
%\includegraphics[width=5.75cm,keepaspectratio=true,angle=0]{./figures/CPSnominal_actual_timeseries_Ex1_crop.pdf}
%%%\includegraphics[width=5.6cm,keepaspectratio=true,angle=0]{./figures/CPSnominal_actual_timeseries-eps-converted-to_crop.pdf}
%\caption{Simulations on the \CPS{} of \autoref{exa:sys}}
%\label{f:HS traj}
%\end{figure*}
%

\enlargethispage{.8\baselineskip}
In \autoref{f:HS traj}, the left graphic collects a campaign of 100
simulations, lasting 250 time units each, showing that the value of the
state variable $\mathit{temp}$ when the cooling system is turned on
(resp., off) lays in the interval $(9.9, 11.5]$ (resp., $(2.9,8.5]$);
these bounds are represented by the dashed horizontal lines. Since
$\delta=0.4$, these results are in line with those of \autoref{prop:X}.
The right graphic shows three examples of possible evolutions in time of
the state variable $\mathit{temp}$.

Now, the reader may wonder whether it is possible to design a variant 
of our engine which meets the same specifications with better 
performances. For instance, an engine consuming less coolant. 
Let us consider a variant of the engine described before: 
\begin{center}
\begin{math}
\overline{\mathit{Eng}} \: = \: \confCPS {\overline{\env}} { \mathit{Ctrl} } \enspace . 
\end{math}
\end{center}
Here,   $\overline{\env}$
  is the same as $\env$ except for the  evolution map, as  we set 
$\mathit{heat}(\actuatorfun^i{},\mathit{cool})=-0.8$ if
$\actuatorfun^i{}(\mathit{cool}) = \on$. This means that in 
$\overline{\mathit{Eng}}$
 we reduce the power  of the cooling system by $20\%$. In 
\autoref{fig:consumption}, we report the results of our simulations 
over $10000$ runs lasting $10000$ time units each. From this graph,  
$\overline{\mathit{Eng}}$ saves in average more than 
$10\%$ of coolant with respect to $\mathit{Eng}$.
 So, the new question is: are these two engines
behavioural equivalent? Do they meet the same specifications? 
%%And in particular, do they respect the same invariants? 

Our  bisimilarity  provides us with 
a precise answer to these  questions. 
\begin{proposition}
\label{prop:performances} The two variants of the engine are bisimilar: 
 $\mathit{Eng} \approx \overline{\mathit{Eng}}$.
\end{proposition}

At this point, one may wonder whether it is possible to improve the performances 
of our engine even more. For instance,  by 
reducing the power of the cooling system 
by a further $10\%$, by setting 
$\mathit{heat}(\actuatorfun^i{},\mathit{cool})=-0.7$ if
$\actuatorfun^i{}(\mathit{cool}) = \on$. We can formally prove that
this is not the case. 
\begin{proposition}
\label{prop:stop} Let $\widehat{\mathit{Eng}}$
 be the  same as  $\mathit{Eng}$, except for the evolution map, 
in which 
$\mathit{heat}(\actuatorfun^i{},\mathit{cool})=-0.7$ if
$\actuatorfun^i{}(\mathit{cool}) = \on$. Then, $\mathit{Eng} \not\approx 
\widehat{\mathit{Eng}}$.
\end{proposition}

\begin{figure}[t]
\centering
\includegraphics[width=7.5cm,keepaspectratio=true,angle=0]{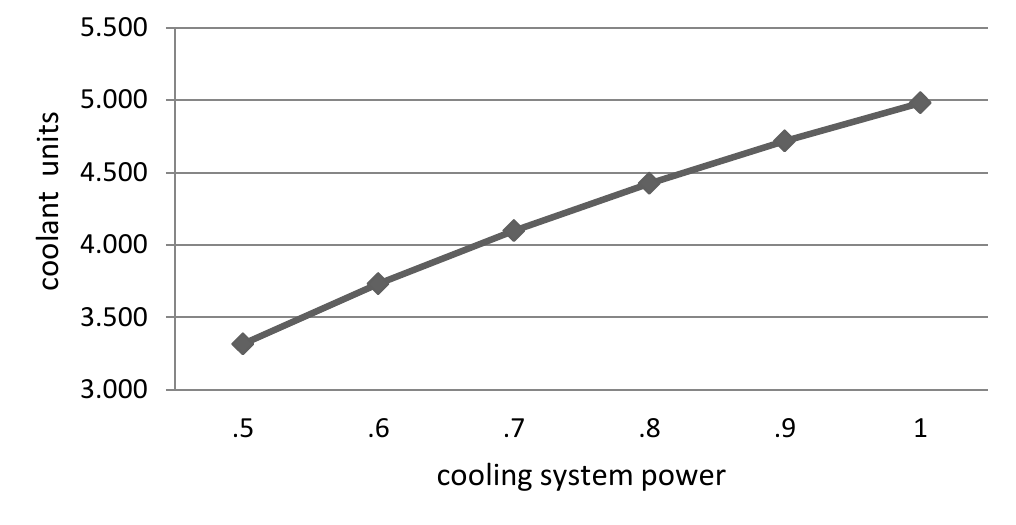}
\caption{Simulations in MATLAB of coolant consumption}
\label{fig:consumption}
\end{figure}

%%\enlargethispage{\baselineskip}
Finally, we show  how we can use the 
compositionality of our 
behavioural semantics  (\autoref{thm:congruence}) to deal with bigger \CPS{s}. 
Suppose that $\mathit{Eng}$ denotes the  modelisation of an airplane
engine. In this case, we could define in \cname{} a very simple 
\emph{airplane control system} that checks whether
the left engine ($\mathit{Eng}_{\mathrm{L}}$) and the right engine
($\mathit{Eng}_{\mathrm{R}}$) are signalling warnings. 
The whole \CPS{} is defined as follows:
\begin{displaymath}
\mathit{Airplane}  \: = \:  \big( ( \mathit{Eng}_{\mathrm{L}} \uplus
\mathit{Eng}_{\mathrm{R}}  )  \parallel \mathit{Check} \big) {\setminus}\{warning \}  
\end{displaymath}%
where $\mathit{Eng}_{\mathrm{L}} = \mathit{Eng}
\subst{\mathrm L}{\mathrm{ID}}
\subst{\mathit{temp{\_}l}}{\mathit{temp}}
\subst{\mathit{cool{\_}l}}{\mathit{cool}} 
\subst{s_{\mathrm{t}{\_}l}}{s_{\mathrm{t}}}$, and, similarly, 
$\mathit{Eng}_{\mathrm{R}} = \mathit{Eng}
\subst{\mathrm R}{\mathrm{ID}}
\subst{\mathit{temp{\_}r}}{\mathit{temp}}
\subst{\mathit{cool{\_}r}}{\mathit{cool}} 
\subst{s_{\mathrm{t}{\_}r}}{s_{\mathrm{t}}}$, 
and process $\mathit{Check}$ is defined as:  
\begin{center}
{\small 
\begin{math}
\begin{array}{rcl}
\mathit{Check} & = &  \fix{X}
\timeout{\LIN{\mathit{warning}}{x }.
\ifelse {x = {\mathrm{L}}} {\mathit{Check}^{\mathrm L}_1}
{\mathit{Check}^{\mathrm R}_1}}{X}\\[3pt]
\mathit{Check}^{\mathit{id}}_i & = & 
\timeout{\LIN{\mathit{warning}}{y }.
\ifelse {y \neq {\mathit{id}}} {\OUTCCS{\mathit{alarm}}.\tick.X}
{\tick. \mathit{Check}^{\mathit{id}}_{i+1}}}
{\mathit{Check}^{\mathit{id}}_{i+1}}\\[3pt]
\mathit{Check}^{\mathit{id}}_5  & = & \lfloor {\LIN{\mathit{warning}}{z }.
\ifelse {z \neq {\mathit{id}}} {\OUTCCS{\mathit{alarm}}.\tick.X}
{ \OUT{\mathit{failure}}{\mathit{id}}.\tick. X}} \rfloor \\
&&
 {\OUT{\mathit{failure}}{\mathit{id}}.X}
%%\\[4pt]
%%\mathit{Check}^{\mathrm Y}_i & = & 
%%\timeout{\LIN{\mathit{warning}}{z }.
%%\ifelse {z \neq {\mathrm{Y}}} {\OUTCCS{\mathit{alarm}}.\tick.X}
%%{\tick. \mathit{Check}^{\mathrm Y}_{i+1}}}
%%{\mathit{Check}^{\mathrm Y}_{i+1}}
\end{array}
\end{math}
}
\end{center}
for $1 \leq i \leq 5$. 
Intuitively, if one of the two engines is in a warning state then the 
process $\mathit{Check}^{\mathit{id}}_i$, for ${\mathit{id}} \in \{ \mathrm{L}, 
\mathrm{R}\}$,   checks whether also the second engine moves into a warning state,  in the following 
 $5$ time intervals (i.e.\ during the cooling cycle). If both engines gets in a 
 warning state then  an $\mathit{alarm}$ is 
sent, otherwise, if only one  engine is facing a  warning then
 the airplane control system yields a \emph{failure} signalling which engine 
is not working properly.

So, since we know that $\mathit{Eng} \approx \overline{\mathit{Eng}}$, the final question becomes the following: can we safely equip our airplane with the  more performant engines, $\overline{\mathit{Eng}_{\mathrm{L}} }$ and $\overline{\mathit{Eng}_{\mathrm{R}} }$, in which 
$\mathit{heat}(\actuatorfun^i{},\mathit{cool})=-0.8$ if
$\actuatorfun^i{}(\mathit{cool}) = \on$,  without affecting the
whole observable 
behaviour of the airplane?
The answer is  ``yes'', and this result can be formally   proved 
by applying \autoref{prop:performances} and \autoref{thm:congruence}. 
\begin{proposition}
\label{prop:air}
Let \begin{math}
\overline{\mathit{Airplane}}  =  \big( (\overline{\mathit{Eng}_{\mathrm{L}} } 
\uplus 
 \overline{\mathit{Eng}_{\mathrm R}}) \parallel \mathit{Check} \big)
{\setminus} \{warning\} 
\end{math}. Then,  
$\mathit{Airplane} \approx \overline{\mathit{Airplane}}$.
\end{proposition}

\section{Related and Future Work}
\label{sec:conclusions}

A number of approaches have been proposed for modelling \CPS{s} using formal methods. 
% Among these, \emph{hybrid automata}~\cite{ACHH1993}
% %%,LSVW1996}
% is a formal model that combines
For instance, \emph{hybrid automata}~\cite{ACHH1993} combine
finite state transition systems with discrete variables (whose values capture the state of the modelled discrete or cyber components) and continuous variables (whose values capture the state of the modelled continuous or physical components). 

\emph{Hybrid process algebras}~\cite{CuRe05}
%%BergMid05,vanBeek06} 
are a powerful tool for reasoning about physical systems and provide techniques for analysing and verifying protocols for hybrid automata.
\cname{} shares some similarities with the 
$\phi$-calculus~\cite{RouSong03}, a hybrid extension of the $\pi$-calculus~\cite{Mil91}. 
In the $\phi$-calculus, a hybrid system is represented as a pair $(E,P)$, where $E$ is
the environment and $P$ is the process interacting with the environment.
Unlike \cname{}, in $\phi$-calculus, given a system $(E,P)$ the process
$P$ can dynamically change  both 
the evolution law and the invariant of the system. However, 
the $\phi$-calculus does not have a representation of
physical devices and measurement law. 
%%This allows the $\phi$-calculus to capture local invariants, i.e.,
%%invariants which hold in a specific state. 
Concerning behavioural semantics, the $\phi$-calculus is equipped with a
weak bisimilarity between systems that is  not compositional.

In the HYPE process algebra~\cite{HYPE}, the continuous part
of the system is represented by  appropriate variables whose changes are 
determined by active influences (i.e., commands on actuators).
The authors defines a strong bisimulation that extends the \emph{ic-bisimulation} of~\cite{BergMid05}. Unlike ic-bisimulation, the
bisimulation in HYPE is preserved by a notion of parallel 
composition that is slightly more
permessive than ours. However, bisimilar systems in HYPE must always have the 
same influence. Thus, in HYPE we cannot compare 
\CPS{s} sending different commands on actuators at the same time, 
as we do in Proposition~\ref{prop:performances}. 
%%\enlargethispage{\baselineskip}

%%\enlargethispage{\baselineskip}
Vigo et al.~\cite{VNN13} proposed a  calculus for wireless-based 
cyber-physical systems 
endowed with a theory to study cryptographic primitives, together with explicit notions of communication failure and unwanted communication. 
 The calculus does not provide any notion of  behavioural equivalence. 
It also lacks a clear distinction between physical and logical 
components. 
%%Compared to~\cite{VNN13}, paper~\cite{WuZh15} 
%%introduces a static network topology and enrich the theory with an harmony theorem.

Lanese et al.~\cite{LBdF13} proposed an untimed 
calculus of mobile IoT devices interacting 
with the physical environment by means of sensors and  actuators. 
The calculus does not allow any representation of the physical environment, 
and the bisimilarity is not preserved by parallel composition (compositionality is recovered by significantly strengthening the discriminating power). 

Lanotte and Merro~\cite{LM16} extended and generalised the work of~\cite{LBdF13}
in a timed setting by providing a %%fully abstract
 bisimulation-based 
semantic theory that is suitable for compositional reasoning. 
As in~\cite{LBdF13}, the  physical environment is not represented.

Bodei et al.~\cite{BDFG16} proposed an untimed 
process calculus supporting  a control flow 
analysis to track how data spread from sensors to the logics of the network, 
and how physical data are manipulated. Sensors and actuators
are modelled as value-passing CCS channels. 
%%Intra-node generative communications in IoT-L{\small Y}S{\small A} are implemented through a shared store \`a la Linda~\cite{Linda}. In this manner physical data are made available to software entities that analyse them and trigger the relevant actuators to perform the desired behaviour. The calculus adopt asynchronous multi-party communication among nodes taking care of node proximity (the topology 
%%is static). 
The dynamics of the calculus is given in terms of a reduction 
relation and no behavioural equivalence is defined. 

As regards future works, 
we believe that our paper can  lay and streamline 
\emph{theoretical foundations} for the development of formal and 
automated tools to verify \CPS{s} before their practical implementation.
 To that end, we will consider applying, possibly after proper enhancements,
 existing tools and frameworks for automated verification, such as 
Maude~\cite{Maude}, Ariadne~\cite{Ariadne}, and SMC UPPAAL~\cite{SMC-Uppaal}, resorting to the development of an dedicated tool if existing ones prove not up to the task.
 Finally, in~\cite{LMMV16}, we developed
an extended version of \cname{} to provide a formal study of a 
variety of  \emph{cyber-physical attacks}  targeting  physical devices.
Also in this case, the final goal is to develop formal and automated tools
to analyse security properties of \CPS{s}.

\vspace*{-1.8mm}
\paragraph*{Acknowledgements} We thank Riccardo Muradore for providing us with simulations in MATLAB.

\bibliography{main}
\bibliographystyle{splncs03}

\appendix

\section{Proofs}

 We recall that the cyber-components
our \CPS{s} are basically TPL-processes~\cite{HR95} extended with  
constructs to read sensors and write actuators. TPL already enjoys 
time determinism, patience and maximal progress. The well-timedness property
is present in many process calculi with a discrete notion of time (e.g.~\cite{MBS11}) similar to ours. Thus, it is 
straightforward to rewrite the proofs 
of those results for our slight variant of TPL. 
\begin{proposition}[Processes time properties~\cite{HR95,MBS11}] 
\label{prop:time2}	
\begin{itemize}
\item 
If $P \trans \tick Q$ and 
$P\trans \tick  {R}$, then   
$Q \equiv  R$. 
\item 
If $P\trans \tau P'$  then there is no $P''$ such that  $P\trans \tick P''$. 

\item 
If $P \trans \tick P'$  for no  $P'$ 
 then  
there is $Q$ such that $P \trans \tau Q$. 

\item
For any $P$  there is a $ k\in\mathbb{N}$ such that  if
$P  \trans{\lambda_1}\dots \trans{\lambda_n} P$, with $\lambda_i \neq \tick$,   then $n\leq k$. 
\end{itemize}
\end{proposition}

The challenge in 
the proof of  \autoref{prop:time} is to lift the results
of \autoref{prop:time2} to the \CPS{s} of \cname.

In its standard formulation, \emph{time determinism} says that a system reaches at most one new state by executing a 
 $\tick$-action.  However,  by an application of Rule~\rulename{Time}, our 
\CPS{s} may nondeterministically move into a new physical environment, according to the evolution law.

\begin{proposition}[Time determinism for \CPS{s}]
\label{prop:timed:aux}
If $M \trans \tick \confCPS {\hat{E}} {Q} $ and 
$M\trans \tick \confCPS {\tilde{E}} {R}$, then $\{ \hat{E}, \tilde{E} \} 
\subseteq \operatorname{next}{(E)}$ and 
$Q \equiv  R$.
\end{proposition}
\begin{proof}
%%We prove that if $M \trans \tick \confCPS {\hat{E}} {Q} $ and 
%%$M\trans \tick \confCPS {\tilde{E}} {R}$, then $\{ \hat{E}, \tilde{E} \} 
%%\subseteq \operatorname{next}{(E)}$ and 
%$Q \equiv  R$.
Let $M=\confCPS { E } {P}$.
Since the only possible   rule to derive $\trans \tick$ is rule \rulename{Time}, then we have that there is $Q$, $R$, $\hat{E}$, $\tilde{E}$ such that 
 \[  \Txiombis
  { P  \trans{\tick} {Q} \q\;
 M   \ntrans{\tau} \q\;
\invariantfun{}(E ) \q\; 
 {\hat{E}} \in \operatorname{next}(E ) }
{M   \trans{\tick} \confCPS {\hat{E}} {Q}}
\]
and
 \[  \Txiombis
  { P  \trans{\tick} {R} \q\;
 M   \ntrans{\tau} \q\;
\invariantfun{}(E ) \q\; 
 {\tilde{E}} \in \operatorname{next}(E ) }
{M   \trans{\tick} \confCPS {\tilde{E}} {R}} \enspace . 
\]
The result follows by  \autoref{prop:time2}.
\end{proof}

According to~\cite{HR95}, the  \emph{maximal progress}
property says that processes communicate as soon as a possibility of communication arises. 
In our calculus,  we generalise this property saying that 
$\tau$-actions  cannot be delayed, independently on how they are
generated. 
\begin{proposition}[Maximal progress for \CPS{s}]
\label{prop:maxprog:aux}
If $M\trans \tau M'$  then there is no $M''$ such that  $M\trans \tick M''$. 
\end{proposition}
\begin{proof}
%%We prove that if $M\trans \tau M'$  then there is no $M''$ such that  $M\trans \tick M''$. 
%%Let us suppose that  $M =\confCPS {E } {P }$. 
The proof is by contradiction.
Let us suppose   $M   \trans{\tick}M''$, for some $M''$.
This   is only possible by an application of   rule \rulename{Time}:
 \[  \Txiombis
  { P  \trans{\tick} {P'} \q\;
 M   \ntrans{\tau} \q\;
\invariantfun{}(E ) \q\; 
 E' \in \operatorname{next}(E ) }
{M   \trans{\tick} M''}
\]
  with   $M''=\confCPS {E'} {P'}$.
However,  the  premises requires  $ M   \ntrans{\tau}$ which contradicts the fact that 
$M\trans \tau M'$.
\end{proof}

Patience in \cname{} is more involved with respect to the same property in
TPL. It basically  says that if a \CPS{} cannot evolve in time, then
either (i) the physical plant does not contemplate an evolution, or (ii)
the invariant is violated, or (iii) the \CPS{} can perform an internal action.  
\begin{proposition}[Patience for \CPS{}]
\label{prop:patience:aux}
If $M \trans \tick M'$  for no  $M'$ 
 then either  $ \operatorname{next}(E )=\emptyset$ or
 $\invariantfun{(M)}=\false$ or
there is $N$ such that $M \trans \tau N$. 
\end{proposition}
\begin{proof}
%%We prove that if $M \trans \tick M'$  for no  $M'$ and   $M =\confCPS {E } {P }$, then 
%%  either $\invariantfun{(E)}=\false$ or
%%there is $N$ such that $M \trans \tau N$  or  
%%%$ \operatorname{next}(E )=\emptyset$. 
%%% MASSIMO: Non capisco perche' ogni volta ripeti lo statement nella prova!
The proof is by contradiction.
Let us suppose that $M \trans \tick M'$  for no  $M'$, and 
 $M   \ntrans{\tau}$  and
$\invariantfun{}(E ) =\true$ and $ E' \in \operatorname{next}(E ) $, for some $E'$.
Since the only possible   rule to derive $\trans \tick$ is rule \rulename{Time},  then  $M \trans \tick M'$  for no  $M'$, implies
that the following derivation is not admissible for any $P'$ and $E'$:
 \[  \Txiombis
  { P  \trans{\tick} {P'} \q\;
 M   \ntrans{\tau} \q\;
\invariantfun{}(E ) \q\; 
 E' {\in} \operatorname{next}(E ) }
{M   \trans{\tick} \confCPS {E'} {P'}} \enspace . 
\]
Since    
 $M   \ntrans{\tau}$  and
$\invariantfun{}(E ) =\true$ and $ E' \in \operatorname{next}(E ) $, for some $E'$,  the only possibility  is $P  \trans{\tick} {P'}$ for no $P'$.
 Since  $    P  \trans{\tick} {P'}$ for no $P'$, by \autoref{prop:time2} we have that 
$   P  \trans{\tau} {P'}$. Since $\invariantfun{}(E ) =\true$, 
by an application of rule  \rulename{Tau}   there is $N$ such that $M \trans \tau N$. This contradicts  the initial hypothesis that $M   \ntrans{\tau}$.
\textcolor{white}{\cite{CHM15,MKN02,LM11}}
\end{proof}

The following property is well-timedness. It basically says that time passing cannot be prevented 
by infinite sequences of internal actions. 
\begin{proposition}[Well-timedness for \CPS{s}]
\label{prop:welltime:aux}
For any $M$  there is a $ k\in\mathbb{N}$ such that  if
$M  \trans{\alpha_1}\dots \trans{\alpha_n} N$, with $\alpha_i \neq \tick$,   then $n\leq k$. 
\end{proposition}
\begin{proof}
The proof is by contradiction. 
%%We prove that for any $M$  there is a $ k\in\mathbb{N}$ such that  if
%%$M  \trans{\alpha_1}\dots \trans{\alpha_n} N$ with $\alpha_i \neq \tick$,   then $n\leq k$. 
Suppose there is no $k$ satisfying the statement above.
Hence there exists an unbounded derivation 
\[
M= M_1  \trans{\alpha_1}\dots \trans{\alpha_n} M_{n+1} \trans{\alpha_{n+1}} \dots
\]
with $M_i= \confCPS {E_i} {P_i}$ and $\alpha_i \neq \tick$.

By inspection of rules of \autoref{tab:lts_systems} we have that, for any $i$,  
$M_i  \trans{\alpha_i} M_{i+1}  $ and $\alpha_i \neq \tick$ implies that $P_i  \trans{\lambda_i} P_{i+1}  $, for some $\lambda_i\neq \tick$.
Hence we have the following unbounded derivation
\[
P_1  \trans{\lambda_1}\dots \trans{\lambda_n} P_{n+1} \trans{\lambda_{n+1}} \dots
\]
with $\lambda_i \neq \tick$. In contradiction with \autoref{prop:time2}.
\end{proof}

\noindent
\textbf{Proof of \autoref{prop:time}}
\begin{proof}
The result follows by an application of \autoref{prop:timed:aux},  \autoref{prop:maxprog:aux},
\autoref{prop:patience:aux} and \autoref{prop:welltime:aux}. 
\end{proof}

In order to prove the compositionality or our bisimilarity, i.e.\ \autoref{thm:congruence}, we divide its statement in three different propositions.

In order to prove that $\approx$ preserves contextuality, we need 
a number of technical lemmas. 
\autoref{lem:aux2}    formalises a number of properties of the compound 
environment $E_1\uplus E_2$. 
\begin{lemma}
\label{lem:aux2}
Let $E_1$ and $E_2$ be two physical environments. If defined, the 
environment $E_1 \uplus E_2$  has the following properties:
\begin{enumerate} 
\item \label{uno}
 $\mathit{read\_sensor}(E_1\uplus E_2,s)  $ is equal to $\mathit{read\_sensor}(E_1 ,s)$, if $s$ is a sensor  of $E_1$, and it 
 is equal to $\mathit{read\_sensor}(E_2 ,s)$, if $s$ is a sensor  of $E_2$;

\item \label{unobis}
$v \in \mathit{read\_sensor}(E_1 ,s)$ implies that $v \in \mathit{read\_sensor}(E_1\uplus E_2,s)  $
for any sensor $s$ in $E_1$ and for any environment $E_2$;

\item \label{due}
$\mathit{update\_act}(E_1\uplus E_2,a,v)  $  is equal to $\mathit{update\_act}(E_1 ,a,v)  \uplus E_2$, if $a$ is an actuator of $E_1$, and it 
 is equal to $E_1 \uplus \mathit{update\_act}(  E_2,a,v)  $, if $a$ is an actuator of $E_2$;

\item \label{duebis}
$\mathit{update\_act}(E_1 ,a,v)  \uplus E_2$  is equal to 
$\mathit{update\_act}(E_1\uplus E_2,a,v)  $   for any actuator $a$ in $E_1$ and for any environment $E_2$;

\item \label{tre}
$\mathit{next}(E_1\uplus E_2) = \{ \, E_1'\uplus E_2'  \,:\,  
E_1' \in \mathit{next}(E_1 )  \mbox{ and }
 E_2' \in \mathit{next}( E_2) \, \}$;
\item \label{quattro}
 $\invariantfun{}(E_1\uplus E_2) =\invariantfun{}(E_1 ) \wedge \invariantfun{}( E_2)$.
\end{enumerate}
\end{lemma}
\begin{proof}
By definition of the operator $\uplus$ on
physical environments.
\end{proof}

\autoref{lem:aux1} serves to propagate untimed actions on parallel 
\CPS{s}. 

\begin{lemma}
\label{lem:aux1}
If $M  \trans{\alpha} M'$, with  $\alpha \neq \tick$,
then $M \uplus O \trans{\alpha} M'\uplus O$, for any non-interfering \CPS{} $O$, with $\invariantfun{}(O)=\true$. 
\end{lemma}
\begin{proof}
The proof is by rule induction on why $M  \trans{\alpha} M'$.
Let us suppose that   
 $M=\confCPS {E_1}  {P_1}$ and   
 $O=\confCPS {E_2}  {P_2}$, for some $E_1$, $E_2$, $P_1$ and $P_2$. 
We can distinguish several cases on why  $M  \trans{\alpha} M'$.
We prove the case in which   $M   \trans{\alpha}M'$  is derived by 
an application of rule \rulename{SensRead}.
The other cases can be proved in a similar manner.
In this case, we have $\alpha=\tau$ and there are 
$s$, $v$, and $P'_1$ such that 
\[
\Txiombis
{P_1  \trans{\rcva s v} P_1'  \Q \invariantfun{}(E_1 ) \Q
\mbox{\small{$v \in \mathit{read\_sensor}(E_1 ,s)$}} 
}
{M   \trans{\tau} M' }
\]
with $M'=\confCPS {E_1 } {P'_1}$.

Since $P_1  \trans{\rcva s v} P'_1$,  by an applicaiton of rule \rulename{Par} we can derive
$P_1 \parallel P_2 \trans{\rcva s v} P'_1 \parallel P_2$.
Since  $\invariantfun{}(E_1 )=\true $ and, by hypothesis,  $\invariantfun{}(E_2 )=\true $, by an application of \autoref{lem:aux2}(\ref{quattro}) we 
derive that $\invariantfun{}(E_1 \uplus E_2 ) =\true$.
Since $v \in \mathit{read\_sensor}(E_1 ,s)$, by an application of \autoref{lem:aux2}(\ref{unobis})
we derive that $v \in \mathit{read\_sensor}(E_1\uplus E_2,s)  $. This is 
enough to derive that:
\[
\Txiombis
{P_1 \parallel P_2 \trans{\rcva s v} P'_1 \parallel P_2 \Q \invariantfun{}(E_1\uplus E_2) \Q
\mbox{\small{$v \in \mathit{read\_sensor}(E_1\uplus E_2,s)$}} 
}
{M \uplus O \trans{\tau}\confCPS {E_1\uplus E_2} {P'_1\parallel P_2} } \enspace .
\]
Hence the result follows by assuming  $M'=\confCPS {E_1 } {P'}$ and $ M' 
\uplus O=\confCPS {E_1\uplus E_2} {P'_1\parallel P_2}$.
\end{proof}

Next lemma says the invariants of bisimilar \CPS{s} must agree. 
\begin{lemma}
\label{lem:bis-inv}
$M \approx N$ implies $\invariantfun{}(M) = \invariantfun{}(N)$.
\end{lemma}
\begin{proof}
The proof is by contradiction. Suppose that $M \approx N$, 
 $\invariantfun{}(M) = \true$ and $ \invariantfun{}(N)=\false$ (the other case is similar).  
 By \autoref{prop:welltime:aux}, there exists 
 a finite derivation $M  \trans{\tau}M_1   \trans{\tau}\dots \trans{\tau} M_n$, with $M_n \ntrans \tau$. 
Since $ \invariantfun{}(N)=\false$, the \CPS{s} $N$ cannot perform any 
action, and in particular, 
 $N \ntrans \tau$. From $M \approx N$  we  derive that 
$M_i \approx N$, for  $1 \leq i \leq n$.  
Since $M_n \ntrans \tau$, by \autoref{prop:maxprog:aux} it follows that  $M_n \trans \tick M'$, for some $M'$.
Since $ \invariantfun{}(N)=\false$ , we have %%$N \ntrans \tau$ and
 $N \ntrans \tick$ and hene also  $N  \not \!\!\!\! \Trans \tick$.

Summarising: $M_n \approx N$,  $M_n \trans \tick M'$ and $N \not \!\!\!\!\Trans \tick$ which contradict the definition of bisimilarity.
\end{proof}

Here comes one of the main technical result: the bisimilarity is
preserved by the parallel composition of non-interfering \CPS{s}. 
\begin{proposition}
\label{lem:cong1}
%%The bisimilarity $\approx$ is preserved by the operator $\uplus$.
 $M \approx N$ implies $M \uplus O \approx N \uplus O$, for any 
non-interfering \CPS{} $O$. 
\end{proposition}
\begin{proof}
%%Let us prove that $\approx$ is preserved by the operator $\uplus$. 
%% 
We show that the relation $\rel=\rel_1 \cup \rel_2$ is a bisimulation where: 
\[
\begin{array}{rcl}
\rel_1 & = & \left \{(M \uplus  O,\; N \uplus  O)  :  M \approx N  \right \} \\[2pt]
\rel_2  & = & \left \{(M, N )  :  \invariantfun{}(M)=\invariantfun{}(N)=\false \right\} \enspace . 
\end{array}
\]
The relation $\rel_2$ is trivially a bisimulation because it contains
pairs of deadlocked \CPS{s}. Thus, we focus on when $(M \uplus  O,\; N \uplus  O) \in \rel_1$.

We proceed by case analysis on why $M \uplus  O \trans{\alpha} \hat{M} $
(the case when  $N \uplus  O \trans{\alpha} \hat{N} $ is symmetric).

\begin{itemize}
\item 
Let  $M \uplus  O \trans{\tau} \hat{M}$, with
 $M=\confCPS {E_1}  {P_1}$ and  $O=\confCPS {E_2}  {P_2}$, 
 for some $E_1$, $E_2$, $P_1$  and $P_2$, by an application 
of rule \rulename{SensRead}. This implies that 
\[
\Txiombis
{P_1 \parallel P_2 \trans{\rcva s v} P'  \Q \invariantfun{}(E_1\uplus E_2) \Q
\mbox{\small{$v \in \mathit{read\_sensor}(E_1\uplus E_2,s)$}} 
}
{M \uplus O \trans{\tau} \hat{M} }
\]
with $ \hat{M}=\confCPS {E_1\uplus E_2} {P'}$. We recall that by definition 
of $\uplus$ the environments $E_1$ and $E_2$ have different physical devices. 
 Thus, there are two cases:
\begin{itemize}
\item $s$ is a sensor of $E_1$.

In this case,  $P_1   \trans{\rcva s v} P'_1 $, for some $P'_1$, and hence  $P'=P'_1\parallel P_2$. Since $\invariantfun{}(E_1\uplus E_2)= \true$ and 
$v \in \mathit{read\_sensor}(E_1\uplus E_2,s)$, by an 
application of \autoref{lem:aux2}(\ref{uno}) and \autoref{lem:aux2}(\ref{quattro}),
we  derive    $\invariantfun{}(E_1 ) =\invariantfun{}(E_2 ) =\true$ and  
  $v \in \mathit{read\_sensor}(E_1,s)$.
Now, let $M'=\confCPS {E_1}  {P_1'}$;  it follows that  $\hat{M} = M' \uplus O$. 
Since  $P_1   \trans{\rcva s v} P'_1 $,  $\invariantfun{}(E_1 )=\true $,    and  
  $v \in \mathit{read\_sensor}(E_1,s)$,
 by an application of rule \rulename{SensRead} we have   $M   \trans{\tau} M'$. 
As $M\approx N$, there is $  N' $ such that   $N\Trans{}  N' $ with $M' \approx N'$.
Since $\invariantfun{}(E_2 )=\true $, by several applications of \autoref{lem:aux1} it follows that $N \uplus O \Trans{}  N' \uplus O = \hat{N} $, 
with $( \hat{M}, \hat{N}) \in \rel_1 \subseteq \rel$.

\item $s$ is a sensor of $E_2$.

In this case,  $P_2   \trans{\rcva s v} P'_2 $, for some $P'_2$, and hence  $P'=P_1\parallel P'_2$. 
Let $O'=\confCPS {E_2}  {P_2'}$;  it follows that  $\hat{M} = M  \uplus O' =
\confCPS {E_1 \uplus E_2}{({P_1} \parallel {P'_2})}$.
Let $N =\confCPS {E_3}  {P_3}$, for some $E_3$ and $P_3$. By an application 
of rule~\rulename{Par}
 we have that $P_3 \parallel P_2 \trans{\rcva s v} P_3 \parallel P_2'$. Since 
$\invariantfun{}(E_1\uplus E_2)$ and 
$v \in \mathit{read\_sensor}(E_1\uplus E_2,s)$,  by an 
application of \autoref{lem:aux2}(\ref{uno}) and \autoref{lem:aux2}(\ref{quattro}), 
we  derive    $\invariantfun{}(E_1 ) =\invariantfun{}(E_2 ) =\true$ and  
  $v \in \mathit{read\_sensor}(E_2,s)$.
As $M \approx N$, by \autoref{lem:bis-inv} it follows that 
 $\invariantfun{}(E_3) =\true$, and hence  $\invariantfun{}(  E_3 \uplus E_2) =\true$.
Since  $v \in \mathit{read\_sensor}(E_2,s)$,  by \autoref{lem:aux2}(\ref{unobis}), it follows that $v \in \mathit{read\_sensor}(E_3 \uplus E_2,s)$.

Summarising  $P_3 \parallel P_2 \trans{\rcva s v} P_3 \parallel P_2'$, 
 $v \in \mathit{read\_sensor}(E_3 \uplus E_2,s)$,  and 
  $\invariantfun{}(  E_3 \uplus E_2) =\true$. Thus,   by an 
application of rule \rulename{SensRead} we have $N \uplus O \trans{\tau}   N \uplus O'$, 
 with $(  M  \uplus O' ,  N \uplus O') \in \rel_1 \subseteq \rel$.  

\end{itemize}

\item  Let  $M \uplus  O \trans{\tau} \hat{M}$, with
 $M=\confCPS {E_1}  {P_1}$ and  $O=\confCPS {E_2}  {P_2}$, 
 for some $E_1$, $E_2$, $P_1$  and $P_2$,
 by an application of rule \rulename{ActWrite}.
This case is similar to the previous ones. Basically we apply  \autoref{lem:aux2}(\ref{due}) instead of   \autoref{lem:aux2}(\ref{uno}),
and  \autoref{lem:aux2}(\ref{duebis}) instead of   \autoref{lem:aux2}(\ref{unobis}).

\item  Let  $M \uplus  O \trans{\tau} \hat{M}$, with
 $M=\confCPS {E_1}  {P_1}$ and  $O=\confCPS {E_2}  {P_2}$, 
 for some $E_1$, $E_2$, $P_1$  and $P_2$, by an  application of
 rule \rulename{Tau}:  
\[
\Txiombis{P_1 \parallel P_2 \trans{\tau} P' \Q \invariantfun{}(E_1 \uplus E_2)}
{ M \uplus O \trans{\tau} \hat{M} }
\]
  with $\hat{M}= \confCPS { E_1 \uplus E_2}  {P'}$. We  can distinguish four cases.

\begin{itemize}
\item  Let $P_1 \parallel P_2 \trans{\tau} P'$ by an application of rule \rulename{Par}, because  $P_1 \trans \tau P_1'$ and $P'=P_1'\parallel P_2$, for some $P'_1$. Since $\invariantfun{}(E_1 \uplus E_2)$, 
by   \autoref{lem:aux2}(\ref{quattro}),
$\invariantfun{}(E_1 )=\invariantfun{}(  E_2) =\true$.
Let $M'=\confCPS {E_1}  {P_1'}$;  we have that  $\hat{M} = M' \uplus O$.
Since  $P_1 \trans \tau P_1'$ and $\invariantfun{}(E_1 ) =\true$,
 by an application of rule \rulename{Tau} we derive   $M   \trans{\tau} M'$. 
As $M\approx N$, there is $  N' $ such that   $N \Trans{}  N' $ with $M' \approx N'$.
Since $\invariantfun{}(E_2 )=\true $, by several applications 
of \autoref{lem:aux1} we have that $N \uplus O \ttrans{} = N' \uplus O = 
\hat{N}$, with  $( \hat{M}, \hat{N}) \in \rel_1 \subseteq \rel$.  

\item  
 Let $P_1 \parallel P_2 \trans{\tau} P'$ by an application of rule \rulename{Par}, because 
$P_2 \trans\tau P_2'$ and $P'=P_1\parallel P_2'$, for some $P'_2$.
Let $O'=\confCPS {E_2}  {P_2'}$;  it follows that  $\hat{M} = M  \uplus O' =
\confCPS {E_1 \uplus E_2}{({P_1} \parallel {P'_2})}$.
Let $N =\confCPS {E_3}  {P_3}$, for some $E_3$ and $P_3$. By an application 
of rule~\rulename{Par}
 we have that $P_3 \parallel P_2 \trans{\tau} P_3 \parallel P_2'$. Since 
$\invariantfun{}(E_1\uplus E_2)$, by an 
application of  \autoref{lem:aux2}(\ref{quattro}), 
we  derive    $\invariantfun{}(E_1 ) =\invariantfun{}(E_2 ) =\true$. 
As $M \approx N$, by \autoref{lem:bis-inv} it follows that 
 $\invariantfun{}(E_3) =\true$, and hence  $\invariantfun{}(  E_3 \uplus E_2) =\true$.

 Summarising
  $P_3 \parallel P_2 \trans{\tau} P_3 \parallel P_2'$ and 
  $\invariantfun{}(  E_3 \uplus E_2) =\true$.
Thus, by an  application of rule \rulename{Tau} we have $N \uplus O \trans{\tau}  N \uplus O'= \hat{N}$, with $( \hat{M}, \hat{N}) \in \rel_1 \subseteq \rel$.

\item Let $P_1 \parallel P_2 \trans{\tau} P'$ by an application of rule \rulename{Com} because  $P_1 \trans {\out c v}    P_1'$ and  $P_2 \trans{\inp c v} P_2'$ and $P'=P_1'\parallel P_2'$, for some $P'_1$ and $P'_2$.
Since $\invariantfun{}(  E_1 \uplus E_2) =\true$, 
by   \autoref{lem:aux2}(\ref{quattro}) follows that 
 $\invariantfun{}(E_1 )=\invariantfun{}(E_2) =\true$.
Let $M'=\confCPS {E_1}  {P_1'}$ and $O'=\confCPS {E_2}  {P_2'}$;  we have that  $\hat{M} = M' \uplus O'$. Since
  $P_1 \trans {\out c v}    P_1'$ and $\invariantfun{}(E_1 ) =\true$,
 by an application of rule \rulename{Out} we have  $M   \trans{\out c v} M'$. 
As $M\approx N$, there are $N_1$, $N_2$ and $N'$ such that 
$N \Trans{} N_1 \trans{\out c v}N_2 \Trans{} N'$,  with $M'\approx N'$. 
Since $\invariantfun{}(E_2 )=\true $, by an appropriate number of applications of  \autoref{lem:aux1} we have that $N \uplus O \Trans{} N_1 \uplus O$. 
Moreover, 
  $  N_1 \trans{\out c v}N_2  $ implies that $P_3\trans{\out c v} P_3'$ for some $P_3$ and $P_3'$ and $E_3$ such that
   $  N_1=\confCPS {E_3}  {P_3}$ and $N_2=\confCPS {E_3}  {P_3'}$ and $\invariantfun{}(E_3 )=\true $.
Since $P_2 \trans{\inp c v} P_2'$
we can use rules \rulename{Com} to derive $P_2 \parallel P_3 \trans{\tau} P_2' \parallel P_3'$.
Moreover from the fact that both  $\invariantfun{}(E_2 ) =\true$ and  $\invariantfun{}(E_3 ) =\true$  
we can derive, by \autoref{lem:aux2}(\ref{quattro}), that $\invariantfun{}(E_3 \uplus E_2 ) =\true$.

Summarising  $P_2 \parallel P_3 \trans{\tau} P_2' \parallel P_3'$ and $\invariantfun{}(E_3 \uplus E_2 )=\true $, and,  for 
 $O=\confCPS {E_2}  {P_2}$ 
and $O'=\confCPS {E_2}  {P_2'}$,
 we can use  rule \rulename{Tau}
to derive $N_1 \uplus O \trans{\tau} N_2 \uplus O'$.
Since $\invariantfun{}(E_2 ) =\true$, 
by an appropriate number of applications of   
\autoref{lem:aux1},  we get  $N_2 \uplus O' \Trans{} N' \uplus O'$.
As 
  $M' \approx N'$, it follows that  $\big(M' \uplus O' \, , \, N' \uplus O' \big) \in \rel_1 \subseteq \rel$. 

\item Let $P_1 \parallel P_2 \trans{\tau} P'$ by an application of rule \rulename{Com} because 
$P_1 \trans{\inp c v} P_1'$ and  $P_2 \trans{\out c v} P_2'$, for some $P'_1$ and $P'_2$.
This case   is similar to the previous one.
\end{itemize}
\end{itemize}

\begin{itemize}

\item Let  $M \uplus O \trans{\tick} \hat{M} $, with
 $M=\confCPS {E_1}  {P_1}$ and  $O=\confCPS {E_2}  {P_2}$ ,  
 for some $E_1$, $E_2$,  $P_1$  and $P_2$. 
This action can be derived only by an application of   rule \rulename{Time}:
 \[  \Txiombis
  { P_1\parallel P_2  \trans{\tick} {P'} \Q
 M \uplus O \ntrans{\tau} \Q
\invariantfun{}(E_1 \uplus E_2) \Q
 E' \in \operatorname{next}(E_1 \uplus E_2) }
{M \uplus O \trans{\tick} \hat{M}}
\]
  with $\hat{M}=\confCPS {E'} {P'}$.
  
The derivation  $P_1\parallel P_2  \trans{\tick} {P'} $ follows by an 
application of rule \rulename{TimePar} because $P_1   \trans{\tick} {P'_1} $ and 
$ P_2  \trans{\tick} {P'_2} $, for some $P_1'$ and $P_2'$, such that  $P'=P_1'\parallel P_2'$.
Since $\invariantfun{}(E_1 \uplus E_2)$, by \autoref{lem:aux2}(\ref{quattro}) 
follows that  $\invariantfun{}(E_1 ) =\invariantfun{}(  E_2) =\true$. 
Since $ E' \in \operatorname{next}(E_1 \uplus E_2)$,  by \autoref{lem:aux2}(\ref{tre}) follows that 
$ E'=E_1' \uplus E_2'$, for some $E_1' \in \operatorname{next}(E_1  ) $ and
$E_2' \in \operatorname{next}(E_1  ) $.
Furthermore, since $ M \uplus O\ntrans{\tau} $, 
by \autoref{lem:aux1} follows that $ M   \ntrans{\tau} $ and $ O \ntrans{\tau} $.
%%Indeed, if  either $ M   \trans{\tau} $ or $ O \trans{\tau} $, by 
%%we would have $ M \uplus O\trans{\tau} $ that is a contradiction.

Thus, since  $P_1   \trans{\tick} P'_1 $, 
 $ M  \ntrans{\tau} $, 
$\invariantfun{}(E_1 ) =\true$, and  $E_1' \in \operatorname{next}(E_1  ) $,
 by an application of rule \rulename{Time} it follows  that  
$M \trans{\tick} M'$, with $M'=\confCPS {E_1'} {P_1'}$.
Similarly, from $P_2  \trans{\tick} {P'_2} $, 
 $ O  \ntrans{\tau} $, 
$\invariantfun{}(E_2) =\true$, and  $E_2' \in \operatorname{next}(E_2 ) $, we can derive that $O \trans{\tick} O'$, 
with  $O'=\confCPS  {E_2'} {P_2'}$. As a consequence, 
  $\hat{M} =M' \uplus O'=\confCPS {E_1' \uplus E_2'} {P_1' \parallel P_2'}$. 

Now, from $M \approx N$ and $M \trans{\tick} M'$, there are $N_1$, $N_2$ and $N'$ such that 
$N \Trans{} N_1 \trans{\tick}N_2 \Trans{} N'$,  with $M'\approx N'$. 
Since $\invariantfun{}(E_2 )=\true $, by an appropriate number of applications of  \autoref{lem:aux1} we have that $N \uplus O \Trans{} N_1 \uplus O$. Next,
 we show that we can apply rule \rulename{Time}
to derive $N_1 \uplus O \trans{\tick} N_2 \uplus O'$. For that we only need to prove
that $N_1 \uplus O \ntrans{\tau}$. We reason by contradiction. 
Since $M \approx N$ and $N \Trans{} N_1$,  
 there is $M_1$ such that $M \Trans{} M_1$, with $M_1 \approx N_1$. 
Since $M \ntrans{\tau}$, it follows that $M=M_1 \approx N_1$. Since $N_1 \trans{\tick} N_2$ and $O \trans{\tick} O'$, by an application of \autoref{prop:maxprog}(b) we can derive 
$N_1 \ntrans{\tau}$ and $O \ntrans{\tau}$. Thus,  $N_1 \uplus O \trans{\tau}$
could be derived only by an application of rule \rulename{Com} where 
$N_1$ interact with $O$, via some channel $c$.
However, as $N_1 \approx M$ the network $M$ could mimic the same interaction
(via the same channel $c$) with $O$, giving rise to a reduction of the form $M \uplus O \Trans{}\trans{\tau}$.  This is in contradiction 
with the initial premises that $M \uplus O \ntrans{\tau}$. Thus, $N_1 \uplus O \ntrans{\tau}$ and by an application of rule \rulename{Time} we can derive $N_1 \uplus O \trans{\tick} N_2 \uplus O'$. It remanins to determine the 
possible evolutions of $ N_2 \uplus O'$. 

There are two cases:
\begin{itemize}
\item The invariant of $O'$ is true.

In this case, by an appropriate number of applications of   
\autoref{lem:aux1} we get  $N_2 \uplus O' \Trans{} N' \uplus O'$.
Thus, 
$N \uplus O \Trans{\tick} N' \uplus O'$, 
with  $(M' \uplus O' \, , \, N' \uplus O') \in \rel_1 \subseteq \rel$, 
because $M' \approx N'$.

\item The invariant of $O'$ is   false.

In this case, 
 we have that $N \uplus O \Trans{\tick} N_2 \uplus O'$, 
with $(M' \uplus O' \, , \, N' \uplus O') \in \rel_2 \subseteq \rel$,  because  $\invariantfun{}(O')=\false$.

\end{itemize}

\item Let $M \uplus O \trans{\inp c v} \hat{M}$, with
 $M=\confCPS {E_1}  {P_1}$ and  $O=\confCPS {E_2}  {P_2}$,
 for some $E_1$, $E_2$, $E'$, $P_1$, $P_2$ and $P'$. 
This derivation can be only due to an application of   rule \rulename{Inp}:
\[
\Txiombis
{P_1 \parallel P_2 \trans{\inp c v}  P' \Q \invariantfun{}(E_1 \uplus E_2)}
{M \uplus O \trans{\inp c v} \hat{M}}
\]
with   $\hat{M}=\confCPS {E_1 \uplus E_2} {P'}$. We distinguish two cases. 
\begin{itemize}
\item $P_1 \trans{\inp c v} P_1'$, for some $P_1'$ such that $P=P_1' \parallel P_2$.

Then, let $M'=\confCPS {E_1}  {P_1'}$;  we have that  $\hat{M} = M' \uplus O$.
Since $\invariantfun{}(E_1 \uplus E_2)$, by   \autoref{lem:aux2}(\ref{quattro}),
it follows that $\invariantfun{}(E_1 )=\invariantfun{}(  E_2)=\true $.
Since $P_1 \trans{\inp c v} P_1'$ and  $\invariantfun{}(E_1 )=\true $, by
 an application of  \rulename{Inp} on $M$ we can derive   $M   \trans{\inp c v} M'$. 
As $M\approx N$, there is $  N' $ such that   $N\ttrans{\inp c v}  N' $ with $M' \approx N'$.
Since $\invariantfun{}(E_2 ) =\true$, by several applications of
 \autoref{lem:aux1} we have that $N \uplus O \ttrans{\inp c v} = N' \uplus O 
=\hat{N} $, with $( \hat{M}, \hat{N}) \in \rel_1 \subseteq \rel$.

\item $P_2 \trans{\inp c v} P_2'$, for some $P_2'$ such that $P=P_1  \parallel P_2'$.

Let $O'=\confCPS {E_2}  {P_2'}$;  we have that  $\hat{M} = M  \uplus O'$.
Since $\invariantfun{}(E_1 \uplus E_2 ) =\true$, 
by \autoref{lem:aux2}(\ref{quattro}) it follows  that    $\invariantfun{}(E_2 )=\true $. Let $N =\confCPS {E_3}  {P_3}$, for some $E_3$ and $P_3$.
Since $M \approx N$ and  $\invariantfun{}(E_2 )=\true $, 
by \autoref{lem:bis-inv} we derive  $\invariantfun{}(E_3)=\true $. 
By \autoref{lem:aux2}(\ref{quattro}) it follows that 
that  $\invariantfun{}(  E_3 \uplus E_2) =\true$.
Furthermore,  by an application of rule \rulename{Par}
 we have  $P_3 \parallel P_2 \trans{\inp c v} P_3 \parallel P_2'$.

 Summarising: 
  $P_3 \parallel P_2 \trans{\inp c v} P_3 \parallel P_2'$ and 
  $\invariantfun{}(  E_3 \uplus E_2) =\true$.
Thus,  by an application of rule \rulename{Inp} we derive $N \uplus O \ttrans{\inp c v} = N \uplus O' =\hat{N}$, with 
 $( \hat{M}, \hat{N}) \in \rel_1 \subseteq \rel$.  

\end{itemize}
\item Let $M \uplus O \trans{\out c v} \hat{M}$. This case is similar to the previous one. 
\end{itemize}
\end{proof}

Now, let us prove the our bisimilarity is preserved by parallel 
composition of non-interfering processes. This is a special 
case of the previous result. 
\begin{proposition}
\label{lem:cong2}
 $M \approx N$ implies $M \parallel P \approx N \parallel P$, for any 
non-interfering 
process $P$.
\end{proposition}
\begin{proof}
We have to prove that  
 $M \approx N$ implies $M \parallel P \approx N \parallel P$, for any 
process $P$ which does not access any physical device.

Let $E_\emptyset$ be the environment with an empty set of state variables, sensors and actuators. It is straightforward to prove that 
 $M \parallel P \approx M\uplus (\confCPS {E_\emptyset}  {P }) $ and $  N \parallel P \approx N\uplus (\confCPS {E_\emptyset}  {P })$.
Since $\approx$ is preserved by the operator $\uplus$, the result follows
by transitivity of $\approx$.
\end{proof}

Finally, we prove that bisimilarity is preserved by channel restriction. 
\begin{proposition}
\label{lem:cong3}
$M \approx N$ implies $M {\setminus} c  \approx N {\setminus} c$, for any 
channel $c$.
 \end{proposition}
\begin{proof}
It is enough to show that the relation 
\( \{ \big( M {\setminus} c  \, , \,   N {\setminus} c \big) :  M \approx N\}
\)
is a bisimulation. The proof proceeds by case analysis on 
why $ M\setminus c \trans{\alpha} \hat{M}$. 
%%
%%\begin{itemize}
%%\setlength{\itemsep}{0pt}
%%\item 
%%Let $\res c M \trans{\alpha} \hat{M}$, for $\alpha \in \{ \tau , \tick , \inp d v , \out d v   \}$ with $d\neq c$. 
%%In this case, this transition has been derived by an application of rule 
%%\rulename{ChnRes} because $M \trans{\alpha} M'$, with $\hat{M} = \res c {M'}$. 
%%%As $M \approx N$ there is $N'$ such that $N \Trans{\alpha} N'$ and $M' \approx N'$. By several applications of rule \rulename{ChnRes} we can derive 
%%$\res c N \Trans{\alpha} \hat{N} = \res c {N'}$, with $(\hat{M} , \hat{N}) \in \rel$. %
%%
%%\item Let $\res c M \trans{\alpha} \hat{M}$, 
%%for $\alpha \in \{ \inp c v, 
%%\out c v \}$. This case is not admissible as rule \rulename{ChnRes}
%%blocks   actions of the that form. 
%%\end{itemize}
%\qed
\end{proof}

\noindent
\textbf{Proof of \autoref{thm:congruence}}
\begin{proof}
By \autoref{lem:cong1}, \autoref{lem:cong2} and \autoref{lem:cong3}.
\end{proof}

In order to prove \autoref{prop:sys} and \autoref{prop:X} 
we use the following lemma that 
formalises the  invariant properties binding 
the state variable $\mathit{temp}$ with the activity of the cooling system.
%%More precisely, the first statement argues on properties holding when the coolant is not in action.

Intuitively,  when the cooling system 
is inactive the value of the state variable $\mathit{temp}$ lays in the 
interval $[0, 11+\epsilon+\delta]$. Furthermore, if the coolant 
is not active and the variable $\mathit{temp}$ lays in the interval 
  $(10+\epsilon, 11+\epsilon+\delta]$ then the cooling will be turned on in the next 
time slot. 
Finally, when active then cooling system will remain so for 
$k\in1..5$   time slots (counting also the current time slot) being the variable $\mathit{temp}$  in the real interval 
 $( 10-\epsilon-k{*}(1{+}\delta) , 11+\epsilon+\delta-k{*}(1{-}\delta)]$. 

\begin{lemma} 
\label{lem:sys}
Let $\mathit{Eng}$ be the system defined in \autoref{sec:case-study}.
Let
\begin{small}
\begin{displaymath}
\mathit{Eng} = \mathit{Eng_1} \trans{t_1}\trans\tick 
\mathit{Eng_2}\trans{t_2}\trans\tick  \dots 
\trans{t_{n-1}}\trans\tick  \mathit{Eng_n}
%%% \trans{t_n}\trans\tick \mathit{Eng_{n+1}}
\end{displaymath}
\end{small}
such that the traces $t_j$ contain no $\tick$-actions, for any $j \in  1 .. n{-}1 $,  and for any  $i \in  1 .. n $ $\mathit{Eng_i}= \confCPS {E_i}{P_i} $ with 
$E_i = \envCPS 
{\statefun^i{}} 
{\actuatorfun^i{}} 
{ \delta }  
{\evolmap{}}
{ \epsilon }  
{\measmap{}}   
{\invariantfun{}}$.
Then, for any $i \in 1 .. n{-}1 $ we have the following:
\begin{enumerate}

\item \label{one}
 if   $ \actuatorfun^i{}(\mathit{cool})= \off $ then
 $\statefun^i{}(\mathit{temp})  \in [0, 11+\epsilon+\delta ]$; 

\item \label{two}
  if   $ \actuatorfun^i{}(\mathit{cool})= \off $ and 
$\statefun^i{}(\mathit{temp})\in (10+\epsilon, 11+\epsilon+\delta ]$ then, in the next time slot,  $\actuatorfun^{i{+}1}{}(\mathit{cool})=\on$;

\item \label{three}
 if  $ \actuatorfun^i{}(\mathit{cool})=\on$ then   $\statefun^i{}(\mathit{temp}) \in ( 10-\epsilon-k {*}(1{+}\delta) , 11+\epsilon+\delta -k{*}(1{-}\delta)] $, 
for some  $k  \in 1 .. 5 $   such that $\actuatorfun^{i-k}{}(\mathit{cool})=\off $ and 
$\actuatorfun^{i-j}{}(\mathit{cool}) =\on $, for $j \in 0..k{-}1$. 
\end{enumerate}
\end{lemma}
\begin{proof}
Let  us denote with  $v_i$  the values  of
the state variable $\mathit{temp}$ in the systems 
 $\mathit{Eng_i}$, i.e., $\statefun^i{}  (\mathit{temp})=v_i $.
Moreover  we will say that  the coolant is active (resp., is not active) in  $\mathit{Eng_i}$ if $\actuatorfun^i{}(\mathit{cool})=\on$
(resp., $\actuatorfun^i{}(\mathit{cool})=\off$).

The proof is  by mathematical induction on $n$, i.e., the 
number $\tick$-actions of our traces. 

The case base $n=1$ follows directly from the definition of $\mathit{Eng}$. 

Let prove the inductive case. 
We assume that the three statements holds for $n-1$ and we prove that they  
also hold for $n$.
\begin{enumerate}[noitemsep]
\item Let us assume that the cooling  is not active  in $\mathit{Eng_{n}}$, 
then we prove that $v_n \in [0, 11+\epsilon +\delta ]$. 

% Hence let us suppose that  the coolant is not active  in $\mathit{Eng_{n}}$ and we prove that $v_n \in [0, 11+\epsilon +\delta ]$.
   We consider separately the   cases in which  the coolant is active or not  in $\mathit{Eng_{n-1}}$
\begin{itemize}[noitemsep]
\item 
 Suppose the coolant is not active  in $\mathit{Eng_{n{-}1}}$ (and inactive in  $\mathit{Eng_{n}}$).

By inductive hypothesis we have 
$v_{n-1} \in [0, 11+\epsilon +\delta ]$. Furthermore,   if   
$v_{n-1} \in (10+\epsilon , 11+\epsilon +\delta ]$ then, by inductive hypotheses,  the coolant must be active   in $\mathit{Eng_{n}}$.
Since we know  in $\mathit{Eng_n}$ the cooling is not active
it follows that 
$v_{n-1} \in [0, 10+\epsilon ]$.
Furthermore, $\mathit{Eng_{n}}$  the temperature
will increase of a value laying in the interval $[1-\delta,1+\delta]=[0.6,1.4]$. Thus $v_{n}$ will be  in 
$ [0.6, 11+\epsilon +\delta ]\subseteq[0, 11+\epsilon +\delta ]$.

\item 
 Suppose the coolant is active  in $\mathit{Eng_{n{-}1}}$ (and  inactive in  $\mathit{Eng_{n}}$).

By inductive hypothesis   $v_{n-1} \in (
10-\epsilon -k *(1+\delta) , 11+\epsilon +\delta -k*(1-\delta)] $ for some $k \in 1..5$ such that the coolant is
not active in $\mathit{Eng_{n{-}1{-}k}}$ and is active in $\mathit{Eng_{n{-}k}},
\ldots, \mathit{Eng_{n-1}}$.

The case $k \in \{1,\ldots,4\}$  is not admissible. 
In fact if  $k \in \{1,\ldots,4\}$ then the coolant would be active for less than $5$ $\tick$-actions as we know that 
%%so  the coolant it will be active in
$\mathit{Eng_{n}}$ is inactive. 
%%contradicting  the fact that the coolant is not active  in $\mathit{Eng_{n}}$ by hypotheses.

Hence it must be $k=5$. Since $\delta=0.4$, $\epsilon=0.1$  and $k=5$, it holds that $v_{n-1 }\in (10-0.1 -5*1.4, 11+ 0.1 +0.4 -5*0.6]=(2.8, 8.6] $. Moreover, since
the coolant is active for $5$ $\tick$ actions, the controller of 
$\mathit{Eng_{n{-}1}}$ checks the
temperature.  However, since $v_{n-1} \in (2.8, 8.6] $ then the coolant is turned off. Thus, in the next time slot,  the temperature
will increase of a value in $[1-\delta,1+\delta]=[0.6,1.4]$. As
a consequence in $\mathit{Eng_{n}}$  we will have $v_{n}
\in [2.8+0,6, 8.6+1.4]=[3.4,10] \subseteq [0, 11+\epsilon +\delta ]$.
\end{itemize}

\item Let us assume that the coolant  is not active  in $\mathit{Eng_{n}}$ and   $v_n \in (10+\epsilon , 11+\epsilon +\delta ]$, then  we prove that 
 the coolant is active  in $\mathit{Eng_{n{+}1}}$.
Since the coolant is not active in
$\mathit{Eng_{n}}$ then it will check the
temperature before the next time slot. Since $v_n \in (10+\epsilon , 11+\epsilon +\delta ]$ and $\epsilon=0.1$, then the
process $\mathit{Ctrl}$ will sense a temperature greater than $10$ and 
the coolant will be turned on. Thus the coolant will be active in
$\mathit{Eng_{n{+}1}}$.

\item Let us assume that the coolant is active in
$\mathit{Eng_{n}}$, then  we prove that  $v_{n} \in ( 10-\epsilon -k *(1+\delta), 11+\epsilon +\delta -k*(1-\delta)] $ for some $k \in 1..5 $
 and  the coolant is not active in $\mathit{Eng_{n{-}k}}$ and  active
in $\mathit{Eng_{n-k+1}}, \dots, \mathit{Eng_{n}}$.

 %%Hence let us suppose that  the coolant is   active  in $\mathit{Eng_{n}}$  and   we prove that
%% there exists  $k \in [1,5] $
%% such that $v_{n} \in ( 10-\epsilon -k *(1+\delta), 11+\epsilon +\delta -k*(1-\delta)] $  the coolant is not active in $\mathit{Eng_{n -k}}$ and is active
%%in $\mathit{Eng_{n-k+1}}, \dots \mathit{Eng_{n}}$.%
%% MAssimo: perche' ripeti la stessa identica cosa 2 volte?

   We separate the    case in which  the coolant is active  in $\mathit{Eng_{n{-}1}}$ from that in which is not active. 

\begin{itemize}[noitemsep]
\item 
 Suppose the coolant is not active  in $\mathit{Eng_{n{-}1}}$ (and active in $\mathit{Eng_{n}}$).

In this case $k=1$ as the coolant is not
active in $\mathit{Eng_{n-1}}$ and it is active in $\mathit{Eng_{n}}$. 
Since  $k=1$, we have to prove $v_n \in (10-\epsilon -(1+\delta), 11+\epsilon +\delta-(1-\delta)]$.

However, since the coolant is not
active in $\mathit{Eng_{n-1}}$ and is active in $\mathit{Eng_{n}}$ it means that the coolant has been switched on in $\mathit{Eng_{n-1}}$ because the sensed temperature  was above $10$ (this may happen
 only if $v_{n-1} > 10-\epsilon $).
By inductive hypothesis, since  the coolant is not active  in $\mathit{Eng_{n-1}}$, we have that
$v_{n-1} \in [0, 11+\epsilon +\delta ]$.
Therefore, from $v_{n-1} > 10-\epsilon $  and 
$v_{n-1} \in [0, 11+\epsilon +\delta ]$ it follows that  $v_{n-1} \in (10-\epsilon , 11+\epsilon +\delta ]$. 
Furthermore, 
since the coolant is active in $\mathit{Eng_{n}}$, the temperature will
decrease of a value in $[1-\delta,1+\delta]$ and therefore
$v_n \in (10-\epsilon -(1+\delta), 11+\epsilon +\delta-(1-\delta)]$  which concludes this case of the proof.

\item 
Suppose the coolant is   active  in $\mathit{Eng_{n{-}1}}$ (and active in $\mathit{Eng_{n}}$ as well).

By inductive hypothesis there is $h \in 1..5$ such that  $v_{n-1} \in (
10-\epsilon -h *(1+\delta) , 11+\epsilon +\delta -h*(1-\delta)] $ and  the coolant is
not active in $\mathit{Eng_{n{-}1{-}h}}$ and is active in $\mathit{Eng_{n{-}h}},
\ldots, \mathit{Eng_{n{-}1}}$.

The case   $h=5$ is not admissible. In fact, since $\delta=0.4$ and $\epsilon=0.1$,
 if $h=5$ then  
 $v_{n-1 }\in (10-0.1 -5*1.4, 11+ 0.1 +\delta -5*0.6]=(2.8, 8.6] $. Furthermore, since
the coolant is already active since $5$ $\tick$ actions, the controller of 
 $\mathit{Eng_{n{-}1}}$ is supposed to check the
temperature. As  $v_{n-1 }\in (2.8, 8.6] $ the coolant 
should be turned off. 
In contradiction  with the the fact that  the coolant is   active  in $\mathit{Eng_{n }}$.

Hence it must be 
$h \in 1 .. 4$. Let us prove that for $k=h+1$ we obtain  our result. 
Namely we have to prove  that, for $k=h+1$,  (i)   $v_{n} \in ( 10-\epsilon -k *(1+\delta), 11+\epsilon +\delta -k*(1-\delta)] $,   and (ii)
 the coolant is not active in $\mathit{Eng_{n{-}k}}$ and  active
in $\mathit{Eng_{n-k+1}}, \dots, \mathit{Eng_{n}}$.

Let us prove the  statement (i). By inductive hypotheses, it holds that  
$v_{n-1} \in ( 10-\epsilon -h *(1+\delta) , 11+\epsilon +\delta -h*(1-\delta)] $.
Since the coolant is active in  $\mathit{Eng_{n}}$
 then
the temperature
will decrease 
Hence,  
$v_{n }   \in ( 10-\epsilon -(h+1) *(1+\delta) , 11+\epsilon +\delta -(h+1)*(1-\delta)]   $. 
Therefore, since  $k=h+1$, we have that 
$v_{n} \in ( 10-\epsilon -k *(1+\delta) , 11+\epsilon +\delta -k*(1-\delta)] $.

 Let us prove the statement (ii).
By
 inductive hypothesis  the coolant is
inactive in $\mathit{Eng_{n-1-h}}$ and  it is active in $\mathit{Eng_{n-h}},
\ldots, \mathit{Eng_{n-1}}$. Now, since the coolant is active in  $\mathit{Eng_{n}}$, for  $k=h+1$, we have that  the coolant is 
not active in $\mathit{Eng_{n-k}}$ and is active in $\mathit{Eng_{n-k+1}},
\ldots, \mathit{Eng_{n}}$ which concludes this case of the proof.
\end{itemize}

\end{enumerate}
\end{proof}

\noindent 
\textbf{Proof of \autoref{prop:sys}}
\begin{proof}
By   \autoref{lem:sys} and since $\delta=0.4$ and $\epsilon=0.1$,
 the value of the state variable $\mathit{temp}$ is always in the real interval 
$[0, 11.5]$. As a consequence, the invariant of the system 
is never violated and the system never deadlocks.
Moreover, after $5$ $\tick$ -actions of cooling the state variable $\mathit{temp}$ is always  in the real interval  $( 10-0.1-5 *1.4 , 11+0.1+0.4-5*0.6]=(2.9, 8.5]$. 
Hence the process $\mathit{Ctrl}$ will never transmit on the channel $\mathit{warning}$.
\end{proof}

\noindent 
\textbf{Proof of \autoref{prop:X}}
\begin{proof}
Let us prove the two statements separately. 
\begin{itemize}
\item   If process $\mathit{Ctrl}$ senses a
temperature above $10$  (and hence $\mathit{Eng}$ turns on the cooling) 
then the value of the state variable
$\mathit{temp}$ is greater than $10-\epsilon $. By  \autoref{lem:sys}
 the value of the state variable $\mathit{temp}$ is always less or equal than
$11+\epsilon +\delta $. Therefore, if $\mathit{Ctrl}$ senses a temperature above $10$,
then the value of the state variable $\mathit{temp}$ is in $(10-\epsilon ,11+\epsilon +\delta ]$.

\item
 By   \autoref{lem:sys}  (third item) the coolant can
 be active for no more than $5$ time slots.
Hence,  by  \autoref{lem:sys}, when 
$\mathit{Eng}$ turns off the cooling system 
the state variable $\mathit{temp}$ ranges over  $( 10-\epsilon -5 *(1+\delta) , 11+\epsilon +\delta-5*(1-\delta)]$.
\end{itemize}
\end{proof}

In order to prove \autoref{prop:performances} 
we use the following lemma that is a variant of \autoref{lem:sys}.
Differently from \autoref{lem:sys}, when active then cooling system will remain so for 
$k\in1..5$   time slots (counting also the current time slot) being the variable $\mathit{temp}$  in the real interval 
 $( 10-\epsilon -k{*}(0.8{+}\delta) , 11.5-k{*}(0.8{-}\delta)]$.

\begin{lemma} 
\label{lem:sys2}
Let $\overline{\mathit{Eng}}$ be the system defined in \autoref{sec:case-study}.
Let
\begin{small}
\begin{displaymath}
\overline{\mathit{Eng}} = \overline{\mathit{Eng_1}} \trans{t_1}\trans\tick 
\overline{\mathit{Eng_2}}\trans{t_2}\trans\tick  \dots 
\trans{t_{n-1}}\trans\tick  \overline{\mathit{Eng_n}}
%%% \trans{t_n}\trans\tick \overline{\mathit{Eng_{n+1}}
\end{displaymath}
\end{small}
such that the traces $t_j$ contain no $\tick$-actions, for any $j \in  1 .. n{-}1 $,  and for any  $i \in  1 .. n $ $\overline{\mathit{Eng_i}}= \confCPS {E_i}{P_i} $ with 
$E_i = \envCPS 
{\statefun^i{}} 
{\actuatorfun^i{}} 
{ \delta }  
{\evolmap{}}
{ \epsilon }  
{\measmap{}}   
{\invariantfun{}}$.
Then, for any $i \in 1 .. n{-}1 $ we have the following:
\begin{enumerate}

\item   if   $ \actuatorfun^i{}(\mathit{cool})= \off $ then
 $\statefun^i{}(\mathit{temp})  \in [0, 11+\epsilon +\delta ]$; 

\item  
  if   $ \actuatorfun^i{}(\mathit{cool})= \off $ and 
$\statefun^i{}(\mathit{temp})\in (10+\epsilon , 11+\epsilon +\delta ]$ then, in the next time slot,  $\actuatorfun^{i{+}1}{}(\mathit{cool})=\on$;

\item  
 if  $ \actuatorfun^i{}(\mathit{cool})=\on$ then   
 $\statefun^i{}(\mathit{temp}) \in ( 10-\epsilon -k {*}(0.8{+}\delta) , 11+\epsilon +\delta -k{*}(0.8{-}\delta)] $, 
for some  $k  \in 1 .. 5 $   such that $\actuatorfun^{i-k}{}(\mathit{cool})=\off $ and 
$\actuatorfun^{i-j}{}(\mathit{cool}) =\on $, for $j \in 0..k{-}1$. 
\end{enumerate}
\end{lemma}
\begin{proof}
Similar to the proof of  \autoref{lem:sys}.
\end{proof}

\noindent 
\textbf{Proof of \autoref{prop:performances}}
\begin{proof}
By \autoref{prop:sys} is sufficient to prove that $\overline{\mathit{Eng}}$ has no trace which deadlocks or emits an alarm.

By   \autoref{lem:sys2} and since $\delta=0.4$ and $\epsilon=0.1$,
 the value of the state variable $\mathit{temp}$ is always in the real interval 
$[0, 11.5]$. As a consequence, the invariant of the system 
is never violated and the system never deadlocks.
Moreover, after $5$ $\tick$ -actions of cooling the state variable $\mathit{temp}$ is always  in the real interval  $( 10-0.1 -5 *1.2 , 11+ 0.1 +0.4-5*0.4]=(3.9, 9.5]$. 
Hence the process $\mathit{Ctrl}$ will never transmit on the channel $\mathit{warning}$.
\end{proof}

\noindent 
\textbf{Proof of \autoref{prop:stop}}
\begin{proof}
 It is is enough to prove that there exists 
an execution trace of the engine $\widehat{\mathit{Eng}}$ containing an 
output along channel $\mathit{warning}$. Then the  result follows by an 
application of \autoref{prop:sys}.

We can easily build up a trace for $\widehat{\mathit{Eng}}$ 
in which, after $10$ $\tick$-actions, in the $11$-th time slot, 
 the value of the state variable  $\mathit{temp}$  is  $10.1 $. 
In fact, it is enough to increase the temperature of $1.01$ degrees
for the first $10$ rounds. Notice that this is an admissible value since,    $1.01  \in [  1-\delta,1+\delta ]= [  0.6,1.4]$.
Being $10.1 $ the value  of the state variable $\mathit{temp}$, there is an execution
 trace  
in which the  sensed temperature is $10$ (recall that $\epsilon=0.1$) and hence 
the cooling system is not activated. However, 
in the following time slot, i.e.\ the $12$-th time slot,
 the temperature may reach at most the value
$10.1  + 1+\delta=11.5$, imposing the activation of the cooling system. 
After $5$ time units of cooling, in the $17$-th time slot, 
the  variable $\mathit{temp}$ will be at most  
$11.5 -5 \ast (0.7-\delta)=11.5-1.5=10$.
Since $\epsilon = 0.1$, the sensed temperature would be  in 
the real interval $[9.9 ,10.1 ] $. Thus, there 
is an execution trace in which the sensed temperature is $ 10.1 $, 
which will be greater than $10$. 
As a consequence,  the warning will be emitted, in the $17$-th time slot.
\end{proof}

\noindent 
\textbf{Proof of \autoref{prop:air}}
\begin{proof}
By \autoref{prop:performances} we derive 
$\mathit{Eng} \approx \overline{\mathit{Eng}}$.
By simple 
$\alpha$-conversion it follows that $\mathit{Eng}_{\mathrm L} \approx \overline{\mathit{Eng}_{\mathrm L}}$ and 
$\mathit{Eng}_{\mathrm R} \approx \overline{\mathit{Eng}_{\mathrm R}}$, respectively. By \autoref{thm:congruence}(1) (and transitivity of $\approx$) it 
follows that $\mathit{Eng}_{\mathrm L} \uplus \mathit{Eng}_{\mathrm R}
\approx \overline{\mathit{Eng}_{\mathrm L}} \uplus \overline{\mathit{Eng}_{\mathrm R}}$. By \autoref{thm:congruence}(2) it follows that 
 $(\mathit{Eng}_{\mathrm L} \uplus \mathit{Eng}_{\mathrm R}) \parallel \mathit{Check}
\approx (\overline{\mathit{Eng}_{\mathrm L}} \uplus \overline{\mathit{Eng}_{\mathrm R}}) \parallel \mathit{Check}$. By \autoref{thm:congruence}(3) we
obtain $\mathit{Airplane} \approx \overline{\mathit{Airplane}}$.
\end{proof}

\end{document}